\documentclass[11pt]{article}
\usepackage{fullpage}
\usepackage{here}

\usepackage{amsthm,amsmath,amssymb}
\usepackage{xcolor}
\usepackage{graphicx}
\usepackage{booktabs}
\usepackage{caption}
\usepackage{hyperref}
\usepackage{enumerate}
\usepackage{enumitem}
\usepackage{lineno}
\usepackage[noend]{algpseudocode}

\makeatletter
\def\BState{\State\hskip-\ALG@thistlm}
\makeatother

\newcommand{\REMOVE}[1]{}%\textcolor[rgb]{.8,0,0}{Remove for short version: #1}}

\newcommand{\eps}{\varepsilon}
\renewcommand{\phi}{\varphi}

\newcommand{\conv}{\textsf{conv}}
\newcommand{\carc}{\textsf{carc}}
\newcommand{\chopWedges}{\texttt{ChopWedges}}
\newcommand{\buildCap}{\texttt{BuildCap}}
\newcommand{\dip}{\texttt{Dip}}
\newcommand{\shearDip}{\texttt{ShearDip}}

\renewcommand{\emph}{\textbf}

\theoremstyle{plain}
\newtheorem{theorem}{Theorem}
\newtheorem{lemma}[theorem]{Lemma}

\newtheorem{definition}[theorem]{Definition}
\newcounter{op}
\newtheorem{operation}[op]{Operation}
\newtheorem{proposition}[theorem]{Proposition}

\renewcommand{\emph}{\textbf}

\title{Circumscribing Polygons and Polygonizations for Disjoint Line Segments\thanks{A preliminary version of this paper appeared in the
\textit{Proceedings of the 35th International Symposium on Computational Geometry}, (SoCG 2019), Portland, OR, {USA}, June 2019, Vol. 129 of LIPIcs, 9:1--9:17.
Research supported in part by the NSF awards CCF-1422311 and CCF-1423615. 
Akitaya was partially supported by NSERC.
Korman was partially supported by MEXT KAKENHI No.~17K12635.}}
\author{
Hugo A. Akitaya\thanks{Department of Computer Science, University of Massachusetts Lowell, Lowell, MA, USA.}
\and
Matias Korman\thanks{Department of Computer Science, Tufts University, Medford, MA, USA.}
\and
Oliver Korten\footnotemark[3]
\and
Mikhail Rudoy\thanks{MIT CSAIL, Cambridge, MA, USA. Now at Google Inc.,  Cambridge, MA, USA.}
\and
Diane L. Souvaine\footnotemark[3]
\and
Csaba D. T\'oth\footnotemark[3]~\thanks{Department of Mathematics, California State University Northridge, Los Angeles, CA, USA.}
}
\date{}

\begin{document}
\maketitle
%\linenumbers

\begin{abstract}
Given a planar straight-line graph $G=(V,E)$ in $\mathbb{R}^2$, a \emph{circumscribing polygon} of $G$ is a simple polygon $P$ whose vertex set is $V$, and every edge in $E$ is either an edge or an internal diagonal of $P$. A circumscribing polygon is a \emph{polygonization} for $G$ if every edge in $E$ is an edge of $P$.

We prove that every arrangement of $n$ disjoint line segments in the plane has a subset of size $\Omega(\sqrt{n})$ that admits a circumscribing polygon, which is the first improvement on this bound in 20 years. We explore relations between circumscribing polygons and other problems in combinatorial geometry, and generalizations to $\mathbb{R}^3$.

We show that it is NP-complete to decide whether a given graph $G$ admits a circumscribing polygon, even if $G$ is 2-regular.
%a perfect matching.
Settling a 30-year old conjecture by Rappaport, we also show that it is NP-complete to determine whether a geometric matching admits a polygonization.
\end{abstract}

\section{Introduction}
\label{sec:intro}

Reconstruction of geometric objects from partial information is a classical problem in computational geometry. In this paper we revisit the problem of reconstructing a simple polygon (alternatively, a triangulated simple polygon) $P$ when some of its edges have been lost. Given a set $V$ of $n$ points in the plane, a polygonization of $V$ is a simple polygon $P$ whose vertex set is $V$. It is easy to see that, unless all points are collinear, $V$ has a simple polygonization. The number of polygonizations is exponential in $n$, and there has been extensive work on determining the minimum and maximum number of polygonizations for $n$ points in general positon as a function of $n$ (see~\cite{garcia2000lower} and~\cite{sharir2013counting} for the latest upper and lower bounds, and~\cite{hurtado2013plane} for a survey on this and related problems).

A natural generalization of this problem is to augment a given planar straight-line graph\footnote{A more accurate term would be {\em plane} straight-line graph, but we decided to use planar straight-line graph as it is a much more commonly used term in the literature.} (PSLG) $G=(V,E)$ into a simple polygon or a Hamiltonian PSLG. In particular, three variants have been considered:
A simple polygon $P$ on vertex set $V$ is a \textbf{polygonization} if every edge in $E$ is an edge of $P$;
 a \textbf{circumscribing polygon} if every edge in $E$ is an edge or an internal diagonal in $P$; and
 a \textbf{compatible Hamiltonian polygon} if every edge in $E$ is an edge, an internal diagonal, or an external diagonal in $P$.

Hoffmann and T\'oth~\cite{HoffmannT03} proved that every plane straight-line matching admits a compatible Hamiltonian polygon, unless all segments are collinear, in which case no such polygon exists. Gr\"unbaum~\cite{Grunbaum94} constructed an arrangement of 6 disjoint segments that does not admit a circumscribing polygon; see Fig.~\ref{fig:diode} (an earlier construction of size 16 is in~\cite{UrabeW92}). However, a circumscribing polygon is known to exist when (i) each segment has at least one endpoint on the boundary of the convex hull~\cite{Mirzaian92}, or (ii) no segment intersects the supporting line of any other segment~\cite{ORourkeR94}. Pach and Rivera-Campo~\cite{PachR98} proved in 1998 that every set of $n$ disjoint segments contains a subset of $\Omega(n^{1/3})$ segments that admits a circumscribing polygon; no nontrivial upper bound is known.

Rappaport~\cite{rappaport1989computing} proved that it is NP-complete to decide whether $G$ can be augmented into a simple polygon. In the reduction, $G$ consists of  disjoint paths, and Rappaport conjectured that the problem remains hard even if $G$ is a perfect matching (i.e., disjoint line segments in the plane). In the special case that $G$ is perfect matching and every segment has at least one endpoint on the boundary of the convex hull, then an $O(n \log n)$ time algorithm can compute a polygonization (or report that none exists~\cite{Rappaport1990}). If $S$ is a set of $n\geq 3$ parallel chords of a circle, then neither $S$ nor any subset of 3 or more segments from $S$ admits a polygonization (so the analogue of the problem of Pach and Rivera-Campo~\cite{PachR98} has a trivial answer in this case). In a related result, Ishaque et al.~\cite{IshaqueST13} proved that $n$ disjoint line segments in general position, where $n$ is even, can be augmented to a 2-regular PSLG (i.e., a union of disjoint simple polygons).

\smallskip\noindent\textbf{Our Results:}
\begin{itemize}%\itemsep -2pt
\item We prove that every set of $n$ disjoint line segments in general position contains a subset of $\Omega(\sqrt{n})$ segments that admit a circumscribing polygon (Theorem~\ref{thm:circumscribe} in Section~\ref{sec:circumscribe}). This is the first improvement over the previous bound of $\Omega(n^{1/3})$~\cite{PachR98} in the last 20 years.
\item While we do not have any nontrivial upper bound for circumscribing polygons proper, we relate that problem to the extensibility of disjoint line segments to disjoint rays. For every $n\in \mathbb{N}$, we construct a set of $n$ disjoint line segments in the plane such that the size of any subset extensible to disjoint rays is $O(\sqrt{n})$ (Section~\ref{sec:rays}).
\item We prove that it is NP-complete to determine whether a given set of disjoint
cycles in the plane admits a circumscribing polygon (Theorem~\ref{thm:hardness2} in Section~\ref{sec:hardness2}).
The reduction is from Hamiltonian paths in 3-connected cubic planar graphs.
\item We prove that it is NP-complete to determine whether a given set of disjoint line segments admits a polygonization (Theorem~\ref{thm:hardness} in Section~\ref{sec:hardness}). This settles a 30-year old conjecture by Rappaport~\cite{rappaport1989computing} in the affirmative.
\end{itemize}
We conclude with a few open problems and three-dimensional generalizations in Section~\ref{sec:con}.

\smallskip\noindent\textbf{Further Related Previous Work.}
Hamiltonicity has fascinated graph theorists and geometers for centuries. Some planar graph results hold for PSLGs as well (i.e., planar graphs with a fixed straight-line embeddings). Hamiltonicity is NP-complete for planar cubic graphs~\cite{GareyJT76}, but can be solved in linear time in 4-connected planar graphs~\cite{CHIBA1989}, and all 4-connected triangulations (i.e., edge-maximal planar graphs) are Hamiltonian~\cite{W31}. In terms of augmentation, a non-Hamiltonian triangulation cannot be augmented to a Hamiltonian planar graph by adding edges or vertices. However, 
%Pach and Wenger~\cite{Pach2001} proved that every planar graph on $n$ vertices can be transformed into a Hamiltonian planar graph on at most $5n$ vertices by subdividing some of the edges, with at most two new vertices per edge, and by adding new edges. 
%
Cardinal et al.~\cite[Theorem~5]{CardinalHKTW18} proved that every planar graph on $n$ vertices can be transformed into a Hamiltonian planar graph by subdividing at most $\lfloor (n-3)/2\rfloor$ edges, with one vertex each, and by adding new edges. 
See also the surveys~\cite{DL10,OZEKI2018} on Hamiltonicity of planar graphs and their applications.

\section{Large Subsets with Circumscribing Polygons}
\label{sec:circumscribe}

For every integer $n\geq 2$, let $f(n)$ be the maximum integer such that every set of $n$ disjoint segments in the plane in general position contains a subset of $f(n)$ segments that admit a circumscribing polygon. Pach and Rivera-Campo~\cite{PachR98} proved that $f(n)=\Omega(n^{1/3})$. By building up on this result, we improve the bound to $\Omega(\sqrt{n})$.

\begin{theorem}\label{thm:circumscribe}
Every set of $n\geq 2$ disjoint line segments in the plane in general position contains a subset of $\Omega(\sqrt{n})$ segments that admit a circumscribing polygon.
\end{theorem}
The remainder of this section is dedicated to proving this statement. We start with a brief overview of our approach for a set $S$ of $n$ line segments.
\begin{itemize}
    \item In Section~\ref{subs_selec} we show how to select a subset of $\Omega(\sqrt{n})$ segments to be included in a circumscribing polygon (cf.\ Fig.~\ref{fig:circum1}). 
    The main property of the selected segments is that they are stabbed by few vertical lines and their slopes are monotonically increasing or decreasing when ordered along their intercepts along those lines. Ultimately, our circumscribing polygon will contain at least a quarter of the selected segments. 
    \item Our initial candidate polygon will be the convex hull of the selected segments. In most cases, though, the convex hull is not a circumscribing polygon. Thus, we introduce four elementary operations, called \chopWedges, \buildCap, \dip, and \shearDip. 
    These operations each make local changes to the candidate solution and increase the number of segment endpoints visited by the solution. 
    %We successively apply these operations to insert more and more segments into our candidate solution.
    In Section~\ref{subs_invar}, we describe all four operations and describe the invariants of the polygons that we maintain in the course of the algorithm.
    \item With the invariants in place, we can describe an algorithm. It successively invokes the four elementary operations and returns a circumscribing polygon for a constant fraction of the  $\Omega(\sqrt{n})$ selected segments. The algorithm proceeds in four phases and is described in Section~\ref{subs_algo}.
\end{itemize}

\subsection{Selecting a Subset of Segments}\label{subs_selec}
Let $S$ be a set of $n\geq 2$ disjoint line segments in the plane. We may assume without loss of generality that none of the segments is vertical, and all segment endpoints have distinct $x$-coordinates. For a subset $S'\subseteq S$, a \emph{halving line}
is a vertical line $\ell$ such that the number of segments in $S'$ strictly contained in the left and right open halfplanes bounded by $\ell$ differ by at most one. In particular, each halfplane contains at most $|S'|/2$ segments from $S'$.

We partition $S$ recursively as follows. Find a halving line $\ell$ for $S$, and recurse on the nonempty subsets of segments lying in each open halfplane determined by $\ell$. Denote by $T$ the recursion tree, which is a binary tree of depth at most $\log n$. We denote by $V(T)$ the set of nodes of $T$, and by $V_i(T)$ the set of nodes at level $i$ of $T$ for $i=0,1,\ldots , \lfloor \log n\rfloor$. Associate each node $v\in V(T)$ to a halving line $\ell_v$ and to the subset $S_v\subseteq S$ of segments that intersect $\ell_v$ without intersecting the halving lines associated with any ancestor of $v$. This defines a partition of $S$ into subsets $S_v$, $v\in V(T)$.

For every $v\in V(T)$, sort the segments in $S_v$ by the $y$-coordinates of their intersections with the line $\ell_v$;
and let $Q_v\subseteq S_v$ be a maximum subset of segments that have monotonically increasing or decreasing slopes.
By the Erd\H{o}s-Szkeres theorem, we have $|Q_v|\geq \sqrt{|S_v|}$ for every $v\in V(T)$. For a refined analysis, we consider the union of the sets $Q_v$ for $v\in V_i(T)$ for $i=0,\ldots, \lfloor \log n\rfloor$, and then take one such union of maximal cardinality.

We need some additional notation. For every $v\in V(T)$, let $n_v=|S_v|$ and $m_v=|Q_v|$. For every integer $i=0,1,\ldots, \lfloor \log n\rfloor$, let $\mathcal{S}_i$ (resp., $\mathcal{Q}_i$) be the union of $S_v$ (resp., $Q_v$) over all vertices $v\in V_i(T)$. Let $\nu_i=|\mathcal{S}_i|$ and $\mu_i=|\mathcal{Q}_i|$. By definition, we have $n=\sum_{i=0}^{\lfloor \log n\rfloor} \nu_i$.

Let $M=\max \{\mu_i:0\leq i\leq \lfloor \log n\rfloor\}$. We claim that
\begin{equation}\label{eq:M}
M\geq \sqrt{n}/2.
\end{equation}
By the Erd\H{o}s-Szekeres Theorem, we have $m_v\geq \sqrt{n_v}$ for every $v\in V(T)$. Since $n_v\leq n/2^i$ for every $v\in V_i(T)$, then  $m_v\geq \sqrt{n_v}=n_v/\sqrt{n_v}\geq n_v/\sqrt{n/2^i}=\sqrt{2^i/n}\cdot n_v$. Summation over all $v\in V_i(T)$ yields  $M\geq \mu_i\geq \sqrt{2^i/n}\cdot \nu_i$, which in turn gives $\nu_i\leq M\sqrt{n/2^i}$.
Summation over all $i=0,\ldots ,\lfloor \log n\rfloor$ now gives
$n=\sum_{i=0}^{\lfloor \log n\rfloor}\nu_i \leq M\sqrt{n}\sum_{i=0}^{\lfloor \log n\rfloor} 2^{-i/2} \leq 2M\sqrt{n}$, hence $M\geq \sqrt{n}/2$, which proves \eqref{eq:M}.

Let $i^*\in \{0,1,\ldots ,\lfloor \log n\rfloor\}$ be an index where $M=\mu_{i^*}$, and put $\widehat{S}_0=\mathcal{Q}_{i^*}$.  By construction, $\widehat{S}_0=\bigcup \{Q_v: v\in V_{i^*}(T)\}$. We further partition $\widehat{S}_0$ into two subsets as follows.
Let $V_{i^*}^<$ (resp., $V_{i^*}^>$) be the set of nodes in $V_i(T)$ such that the slopes in $Q_v$ monotonically increase (resp., decrease). Let $\widehat{S}_1$ be the larger of $\bigcup \{Q_v: v\in V_{i^*}^<\}$ and $\bigcup \{Q_v: v\in V_{i^*}^>\}$, breaking ties arbitrarily. Note that $|\widehat{S}_1|\geq \sqrt{n}/4$. We may assume, by a reflection in the $y$-axis if necessary, that $\widehat{S}_1=\bigcup \{Q_v: v\in V_{i^*}^>\}$; see Fig.~\ref{fig:circum1} for an example.

\subsection{Algorithm Invariants and Elementary Operations}\label{subs_invar}
%\smallskip\noindent\textbf{Construction of a Circumscribing Polygon.}
Pach and Rivera-Campo~\cite{PachR98} proved that an arrangement of disjoint line segments admits a circumscribing polygon if they are (1) stabbed by a vertical line, and (2) have monotonically increasing or decreasing slopes (in particular, each $Q_v$, $v\in V(T)$, admits a circumscribing polygon). 

In contrast, we construct a circumscribing polygon for the \emph{union} of all $Q_v$, $v\in V_{i^*}(T)$, separated by vertical lines. Note that our construction will be a circumscribing polygon for a large subset $\widehat{S}_2\subseteq \widehat{S}_1$ (not the whole set $\widehat{S}_1$). 

\begin{figure}[htbp]
	\centering
	\includegraphics[width=.9\linewidth]{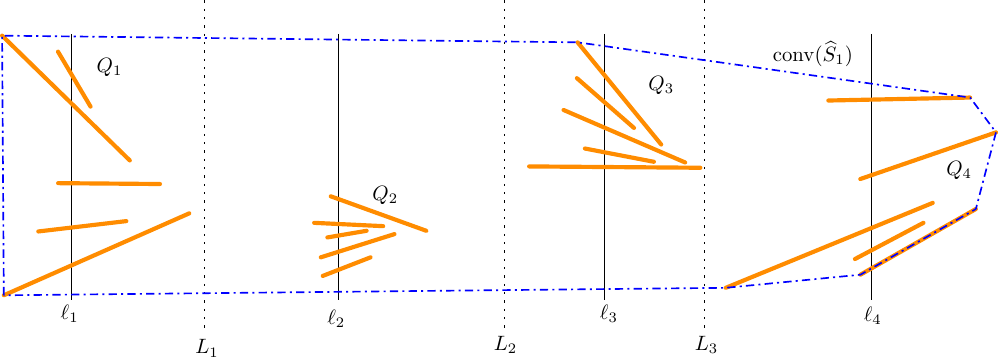}
	\caption{A set $\widehat{S}_1=\bigcup \{Q_v: v\in V_{i^*}^>\}$ of 25 line segments for $r=4$; and $P={\rm conv}(\widehat{S}_1)$.}
	\label{fig:circum1}
\end{figure}

For ease of presentation, we introduce new notation for $\widehat{S}_1=\bigcup \{Q_v: v\in V_{i^*}^<\}$; see Fig.~\ref{fig:circum1}.
Denote by $\ell_1,\ell_2,\ldots , \ell_r$ the halving lines $\{\ell_v: v\in V_{i^*}^<\}$ sorted from left-to-right, and let $Q_1,Q_2,\ldots, Q_r$ be the corresponding sets in $\{Q_v: v\in V_{i^*}^>\}$. Denote by $L_i$ ($i=1,\ldots ,r-1$) the vertical lines that separate $Q_i$ and $Q_{i+1}$.
Refer to Fig.~\ref{fig:circum1}.

\paragraph{Overview.}
We construct a circumscribing polygon for a subset $\widehat{S}_2\subseteq \widehat{S}_1$ incrementally,
while maintaining a polygon in the segment endpoint visibility graph of $\widehat{S}_1$. We use a machinery developed in~\cite{HoffmannT03}, with several important new elements.
Initially, let $P$ be the boundary of $\conv(\widehat{S}_1)$, where \conv\ denotes the convex hull. 
Intuitively, think of polygon $P$ as a rubber band, and stretch it successively to visit more segment endpoints from $\widehat{S}_1$, maintaining the property that all segments in $\widehat{S}_1$ remain in the closed polygonal domain of $P$.  
A key invariant of $P$ will be that if $P$ visits only one endpoint of some segment in $\widehat{S}_1$, then we can stretch it to visit the other endpoint (a strategy previously used in~\cite{HoffmannT03,Mirzaian92}). 
This tool allows us to produce a circumscribing polygon for a subset of $\widehat{S}_1$. 
We ensure that $P$ reaches an endpoint of \emph{at least a quarter} of the segments in $\widehat{S}_1$. To do this, we use the fact that each set $Q_v$, $v\in V_{i^*}^>$, is sorted along the halving lines in decreasing order by slope, and we ensure that $P$ reaches the left endpoint of at least a quarter of the segments (later, we stretch $P$ to visit the right endpoints). At the end, we define $\widehat{S}_2$ to be the set of segments in $\widehat{S}_1$ visited by $P$ (i.e., we discard the remaining segments lying in the interior of $P$).

We maintain a polygon with the properties listed in Definition~\ref{def:frame} below. There are a few important features to note:
$P$ is not necessarily a simple polygon in intermediate steps of the algorithm: it may be a weakly simple polygon that does not have self-crossings; it has clearly defined interior and exterior; and it can have repeated vertices.  Specifically, each vertex can repeat at most twice (i.e., multiplicity at most 2), and if its multiplicity is 2, then one occurrence is a reflex vertex and the other is convex. Furthermore, all such reflex vertices can be removed simultaneously by suitable shortcuts (cf.~property \ref{F5} below) to obtain a simple polygon. We need to be very careful about reflex vertices in $P$: for each reflex vertex in $P$, we ensure either that it will not become a repeated vertex later, or that if it becomes a repeated vertex, its reflex occurrence can be removed by a suitable shortcut.

\paragraph{Invariants.}
As in~\cite{HoffmannT03}, we maintain a weakly simple polygon, called a \emph{frame} (defined below).
A weakly simple polygon is a closed polygonal chain $P=(v_1,\ldots , v_k)$ in counterclockwise order such that, for every $\eps>0$, displacing the vertices by at most $\eps$ can produce a simple polygon. Denote by $\widehat{P}$ the union of the interior and the boundary of $P$.
A weakly simple polygon may have repeated vertices. Three consecutive vertices $(v_{i-1},v_i,v_{i+1})$ define an interior angle $\angle(v_{i-1},v_i,v_{i+1})$, or $\angle v_i$, which is either convex (less than or equal to $180^\circ$) or reflex (more than $180^\circ$).

The following definitions summarizes the properties that we maintain for a polygon $P$. It is based on a similar concept in~\cite{HoffmannT03}: we do not allow segments to be external diagonals (cf.~\ref{F2}) and relax the conditions on the possible occurrences of reflex vertices. 
%Property~\ref{F6} is related to the vertical lines $\ell_i$ \hugo{$L_i$} ($i=1,\ldots ,r-1$) in the instance $S=\bigcup_{i=1}^r Q_i$. 
Reflex vertices play an important role. We distinguish two types of reflex vertices: A reflex vertex $v$ of a frame $P$ is \textbf{safe} if the (unique) line segment in $S$ incident to $v$ subdivides the reflex angle $\angle v$ into two convex angles; otherwise $v$ is \textbf{unsafe} (see Fig.~\ref{fig:frame} for examples).

\begin{definition}\label{def:frame}
A weakly simple polygon $P=(v_1,\ldots, v_k)$ is called \emph{frame} for a set $S$ of disjoint line segments in the plane, if
\begin{enumerate}[label={\rm (F\arabic*)},series=F]\itemsep 0pt
\item\label{F1} every vertex of $P$ is an endpoint of some segment in $S$;
\item\label{F2} $\widehat{P}$, the union of the interior and the boundary of $P$, contains every segment in $S$;
\item\label{F3} every vertex of $P$ has multiplicity at most 2;
\item\label{F4} if a vertex of $P$ has multiplicity 2, say $v_i=v_j$, then one of $\angle v_i$ or $\angle v_j$ is convex (and the other is reflex);
\item\label{F5} if $(v_i,\ldots , v_j)$ is a maximal chain of reflex vertices of $P$ that each have multiplicity 2, then $(v_{j+1},v_j,\ldots v_i,v_{i-1})$ is a simple polygon that is disjoint from the interior of $P$ (cf.~Fig.~\ref{fig:frame});
%\item\label{F6} the vertical line $\ell_i$ ($i=1,\ldots , r-1$) crosses $P$ exactly twice: once in the upper arc and once in the lower arc.
\end{enumerate}
\end{definition}

\begin{figure}[htbp]
	\centering
	\includegraphics[width=.55\linewidth]{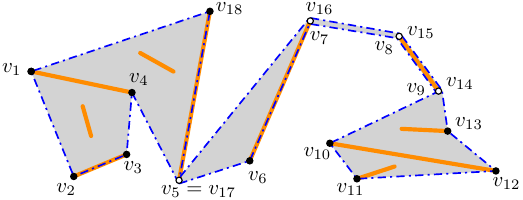}
	\caption{A frame $P=(v_1,\ldots, v_{18})$ for 10 disjoint line segments (orange). The closed region $\widehat{P}$ is shaded gray. The vertices of multiplicity~1 (resp., 2) are marked with full (resp., empty) dots. The reflex vertex $v_{13}$ is safe, all other reflex vertices ($v_4$, $v_7$, $v_8$, $v_9$, and $v_{17}$) are unsafe.}
	\label{fig:frame}
\end{figure}

\paragraph{Elementary Operations.}
Let $S$ be a set of disjoint line segments  in general position, and let $P$ be a frame. %(cf.~Definition~\ref{def:frame}).
For $S$ we define four elementary operations that each transform $P$ into a new frame. The first operation is the ``shortcut'' that eliminates reflex vertices of multiplicity 2, and increases the area of the interior.
The remaining three operations each increase the number of vertices of the frame (possibly creating vertices of multiplicity 2) and decrease the area of its interior.
For shortest path and ray shooting computations, we consider the line segments in $S$ and the current frame $P$ to be obstacles.
For a polygonal path $(a,b,c)$ that does not cross any segment in $S$, we define the \emph{convex arc} $\carc(a,b,c)$ to be the shortest polygonal path between $a$ and $c$ that is homotopic to $(a,b,c)$.

\begin{figure}[htbp]
	\centering
	\includegraphics[width=.95\linewidth]{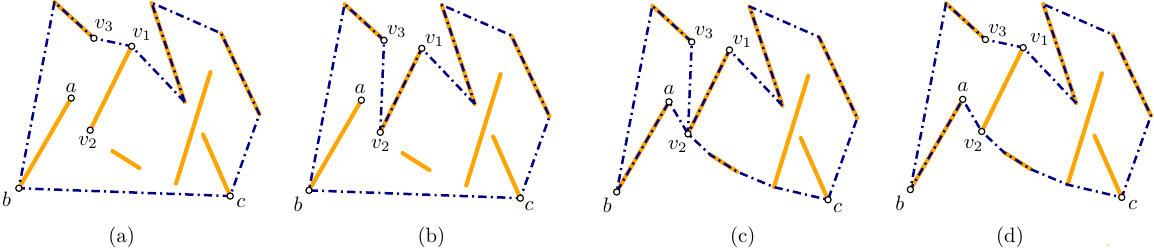}
	\caption{(a) A frame $P$.
    (b) $P:={\rm \buildCap}(P,1,v_1)$.
    (c) $P:={\rm \buildCap}(P,1,b)$.
    (d) $P:={\rm \chopWedges}(P)$.} 
	\label{fig:operations12}
\end{figure}

\begin{operation}\label{def:Chop}
\emph{(\chopWedges$(P)$)} Refer to Fig.~\ref{fig:operations12}(c-d).
\emph{Input:} a frame $P$.
\emph{Action:} While there is a vertex of multiplicity 2, do:
let $(v_i,\ldots , v_j)$ be a maximal chain of reflex vertices of $P$ that each have multiplicity 2,
and replace the path $(v_{i-1},v_i,\ldots v_j,v_{j+1})$ in $P$ by a single edge $v_{i-1}v_{j+1}$.
\end{operation}

\begin{operation}\label{def:BuildCap}
\emph{(\buildCap$(P,\varrho,b)$)} Refer to Fig.~\ref{fig:operations12}(a--c)
\emph{Input:} a frame $P$, an orientation $\varrho\in \{-1,+1\}$, and vertex $b$ is of multiplicity 1 in $P$ such that $b$ the endpoint of $ab \in S$, where $a$ is not a vertex of $P$. Let $c$ be the neighbor of $b$ in polygon $P$ in orientation $\varrho$ (where ccw$=1$, cw$=-1$), and assume $\angle abc$ is convex.
\emph{Action:}  Replace the edge $bc$ of $P$ with the polygonal path $ba+\carc(a,b,c)$. 
\end{operation}

\begin{figure}[htbp]
	\centering
	\includegraphics[width=.95\linewidth]{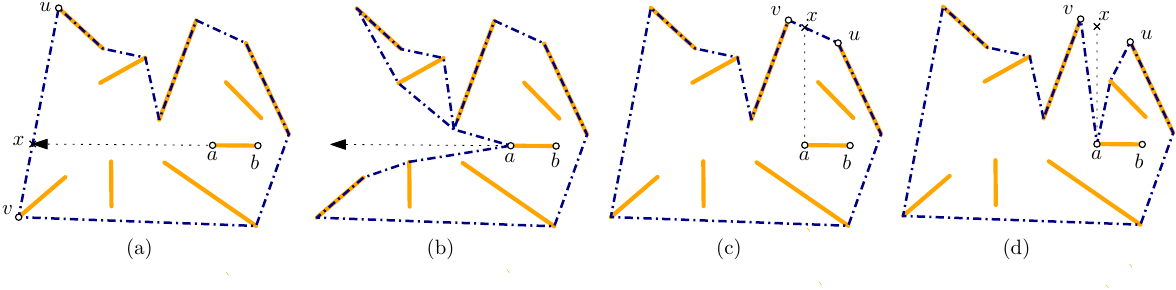}
	\caption{(a) A frame. (b) The result of operation \dip$(P,a,b)$.
	(c)  A frame. (d) The result of operation \shearDip$(P,a,x)$.}
	\label{fig:operations34}
\end{figure}

\begin{operation}\label{def:Dip}
\emph{(\dip$(P,a,b)$)}  Refer to Fig.~\ref{fig:operations34}(a-b).
\emph{Input:} a frame $P$, a segment $ab\in S$ such that neither $a$ nor $b$ is a vertex of $P$, and
the ray $\overrightarrow{ba}$ hits an edge of $P$ that is not a segment in $S$.
\emph{Action:} Assume that $\overrightarrow{ba}$ hits edge $uv$ of $P$ at the point $x$. Replace the edge $uv$ of $P$ with the polygonal path $\carc(u,x,a)+\carc(a,x,v)$.
\end{operation}

\begin{operation}\label{def:ShearDip}
\emph{(\shearDip$(P,a,x)$)}  Refer to Fig.~\ref{fig:operations34}(c-d).
\emph{Input:} a frame $P$, a segment endpoint $a$ in the interior of $P$, and a point $x$ in the interior of an edge of $P$ that is not a segment in $S$ such that $ax$ does not cross $P$ or any segment in $S$.
\emph{Action:} Let $uv$ be the edge of $P$ that contains $x$. Replace the edge $uv$ of $P$ with the polygonal path $\carc(u,x,a)+\carc(a,x,v)$.
\end{operation}

Operations~1 and~2 have been previously used in~\cite{HoffmannT03,Mirzaian92}; it was shown that if \emph{all} reflex vertices in $P$ have been created by \buildCap\ operations, then \ref{F5} automatically holds~\cite[Sec.~ 2]{HoffmannT03}. It is also not difficult to see that if \emph{all} reflex vertices in $P$ have been created by \dip\ operations, then \ref{F5} holds. 
However, this property does not extend to a mixed sequence of \buildCap\ and \dip\ operations, and certainly not for \shearDip\ operations. 
We maintain \ref{F5} by a careful application of these operations, using the fact that each set $Q_i$ ($i=1,\ldots, r$) is stabbed by a vertical line.

Note that Operations 1--4 can only increase the vertex set of the frame (the \chopWedges\ operation decreases the multiplicity of repeated vertices from 2 to 1, but maintains the same vertex set). Initially, $P=\partial \conv(S)$, and so all vertices of $\conv(S)$ remain vertices in $P$ in the course of the algorithm. In particular, the leftmost and rightmost segment endpoint in $S$ are always vertices in $P$ (with multiplicity 1 by \ref{F4}). These vertices subdivide $P$ into an \emph{upper arc} and a \emph{lower arc}. As a convention, the leftmost (resp., rightmost) vertex is part of the lower arc (upper arc).  
When our algorithm invokes the \buildCap\ operation at a vertex $v$, we use \buildCap$(P,\varrho(v),v)$.

We can now justify the distinction between safe and unsafe reflex vertices.
\begin{lemma}\label{lem:safe}
Let $v$ be a reflex vertex of multiplicity 1 in a frame $P$ such that $v$ is safe.
Then after any sequence of the above four operations, 
$v$ becomes a convex vertex or remains a safe reflex vertex of the frame,
in both cases the multiplicity of $v$ remains 1.
\end{lemma}
\begin{proof}
Each operation creates at most one new reflex vertex, which has multiplicity 1; and possibly many convex vertices along the convex arcs, which may have multiplicity 1 or 2. However, each point can be an interior vertex of at most one convex arc.
Consequently, the multiplicity of a vertex $v$ can possibly increase from 1 to 2 if it is first a reflex vertex of multiplicity 1, and then visited for a second time (by a convex arc) as a convex vertex; see Fig.~\ref{fig:operations12}(c) and Fig.~\ref{fig:operations34}(d) for examples. If $v$ is a \emph{safe} reflex vertex, then it cannot be an interior vertex of a convex arc, and so its multiplicity cannot increase from 1 to 2.
Since each operation decreases the interior of the frame, the interior angle of the frame at $v$ can only decrease. Consequently, $v$ either becomes a convex vertex or remains a safe reflex vertex. 
 \end{proof}
 
\paragraph{Additional Invariants for Vertical Lines.} 
In this section, we assume that our instance $S$ is the set $\widehat{S}_1=\bigcup_{i=1}^r Q_i$ constructed in Section~\ref{subs_selec}.
The first two phases of our algorithm will maintain the following additional properties \ref{F6}--\ref{F7} for frames of $\widehat{S}_1$. 
We make this distinction because we use Definition~\ref{def:frame} and the elementary operations introduced here for a different class of input in Section~\ref{sec:rays}.

\begin{enumerate}[resume*=F]\itemsep 0pt
%[label={\rm (F\arabic*)},series=F,resume]\itemsep 0pt
\item\label{F6} The vertical line $L_i$, $i\in \{1,\ldots , r-1\}$, crosses $P$ exactly twice: once in the upper arc and once in the lower arc.
\item\label{F7} The line $\ell_i$, $i\in \{1,\ldots , r\}$, crosses the lower arc (resp., upper arc) of $P$ an odd number of times, and the $y$-coordinates of the intersection points with the lower arc (resp., upper arc) are monotonically decreasing in a counterclockwise traversal of $P$. 
\end{enumerate}

\begin{figure}[htbp]
	\centering
	\includegraphics[width=.9\linewidth]{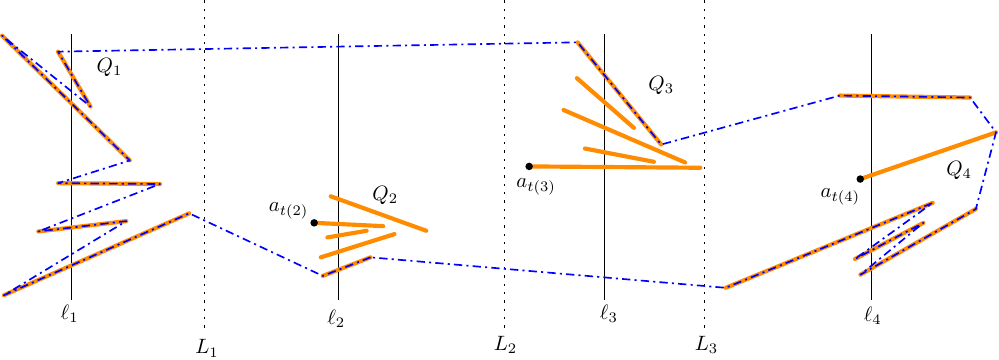}
	\caption{A set $\widehat{S}_1$ of 25 line segments and the frame $P$ at the end of Phase~1.}
	\label{fig:circum2}
\end{figure}

\subsection{Algorithm Description}\label{subs_algo}
With the operations and invariants described above we can proceed to describe our algorithm. Our algorithm has four phases described below. We define a default orientation for every (current and future) vertex $v$ of the frame $P$: if $v$ is in the lower arc, then $\varrho(v)=1$, otherwise $\varrho(v)=-1$. (We shall make some exceptions to this default assignment after Phase~1.)   

\paragraph{Phase~1: Left Endpoints.}
Initially, the frame $P$ is the boundary of the convex hull $\conv(\widehat{S}_1)$. In the first phase of our algorithm, we use operations \buildCap\  and \dip\   as follows:

\begin{enumerate}%\itemsep -2pt
\item Initialize $P:=\partial \conv(\widehat{S}_1)$.
\item While condition (a) or (b) below is applicable, do:
    \begin{enumerate}
    \item If there exists a segment $ab\in S$ such that the left endpoint $a$ is a vertex of $P$, but the right endpoint is not, then set $P:={\buildCap}(P,\varrho(a),a)$.
    \item Else if there exists a segment $ab\in Q_i$ for some $i\in \{1,\ldots, r\}$ such that $a$ is the left endpoint, $a$ lies in the interior of $P$, and $\overrightarrow{ba}$ hits an edge $uv$ where $uv\not\in \widehat{S}_1$, and the left endpoint of $uv$ is an endpoint of some segment in $Q_i$, then set $P:={\dip}(P,a,b)$.
   \end{enumerate}
    \item Return $P$, and terminate Phase~1.
\end{enumerate}

An example is shown in Fig.~\ref{fig:circum2}. 
First we show that Phase~1 returns a frame. Instead of addressing property~\ref{F5} directly, we establish the stronger property~\ref{R1} defined below.

\begin{lemma}\label{lem:ph1correctness}
All operations in Phase~1 maintain \ref{F1}--\ref{F7} and the following property.
\begin{enumerate}[label={\rm (R\arabic*)},series=R]
    \item\label{R1} If $v_i$ is an unsafe reflex vertex of $P$,  then 
    $v_i$ is the right endpoint of a segment in $\widehat{S}_1$, vertices $v_{i-1}$ and $v_{i+1}$ are convex and the triangle $(v_{i+1},v_i,v_{i-1})$ is disjoint from the interior of $P$.
\end{enumerate}
\end{lemma}

\begin{proof}
Phase~1 applies a sequence of \buildCap\ and \dip\ operations. Both operations automatically maintain \ref{F1} and \ref{F2}. Both operations create at most one new reflex vertex in $P$ (at vertex $a$ and $b$, respectively). In both cases, the reflex vertex has not been previously a vertex, so every reflex vertex has multiplicity 1 at the time it is created. A reflex vertex may be visited in a subsequent operation as a convex vertex, but at most once. This establishes \ref{F3} and \ref{F4}.

Each operation $\buildCap(P,\varrho(a),a)$ maintains \ref{F6}: it stretches $P$ from one endpoint to the other endpoint of a segment $ab\in \widehat{S}_1$; both endpoints lie between two consecutive lines $L_{i-1}$ and $L_{i}$.
We show that \ref{F6} is also maintained by each operation ${\dip}(P,a,b)$. To see this, recall that $a$ is the left endpoint of a segment $ab\in Q_i$, and so the point $x$ where ray $\overrightarrow{ba}$ hits the frame $P$ is to the left of $a$.
The algorithm applies operation $\dip(P,a,b)$ if $x$ lies on an edge $uv$ of $P$, and the left endpoint of $uv$, say $u$, is an endpoint of some segment in $Q_i$. Therefore, $a$, $b$, and $x$ are all between the lines $L_{i-1}$ and $L_i$. Thus, $\carc(u,x,a)$ does not cross these lines. If $v$ is between $L_{i-1}$ and $L_i$, then so is $\carc(a,x,v)$; otherwise $\carc(a,x,v)$ crosses the same vertical lines $L_j$, $j\geq i$, that edge $uv$ crossed before the operation, each at most once.

Each operation $\buildCap(P,\varrho(a),a)$ maintains \ref{F7}: Assume the operation introduces a new pair of crossings with the vertical line $\ell_i$ at $ab\cap \ell_i$ and $\carc(a,b,c)\cap \ell_i$, where $c$ is the neighbor of $a$ in orientation $\varrho(a)$. If $a$ is on the lower arc, then $\varrho(a)=1$ by convention, and  $\carc(a,b,c)\cap \ell_i$ is below  $ab\cap \ell_i$. Otherwise $a$ is on the upper arc, then $\varrho(a)=-1$ and $\carc(a,b,c)\cap \ell_i$ is above $ab\cap \ell_i$. In both cases, \ref{F7} is maintained. 

We show that \ref{R1} implies \ref{F5}. Note that initially, $P$ has no reflex vertices, so all reflex vertices are created by an operation in Phase~1. By Lemma~\ref{lem:safe}, the multiplicity of safe vertices of $P$ remains 1. 
Let $(v_i,\ldots , v_j)$ be a maximal chain of reflex vertices that each have multiplicity 2 in $P$. By \ref{R1}, $i=j$, and the chain consists of a single unsafe reflex vertex $v_i$. 
\ref{R1} further implies that $(v_{i+1},v_i,v_{i-1})$ is a simple polygon disjoint from the interior of $P$.

It remains to show that all operations in Phase~1 maintain  \ref{R1}. 
Each operation creates at most one new reflex vertex. First, consider a reflex vertex created by an operation ${\dip}(P,a,b)$. It creates a reflex vertex at $a$, which is safe by construction, and it never becomes an unsafe reflex vertex by Lemma~\ref{lem:safe}. It follows that every unsafe reflex vertex of $P$ in the course of Phase~1 has been created by a $\buildCap$ operation.

\begin{figure}[htbp]
	\centering
	\includegraphics[width=.95\linewidth]{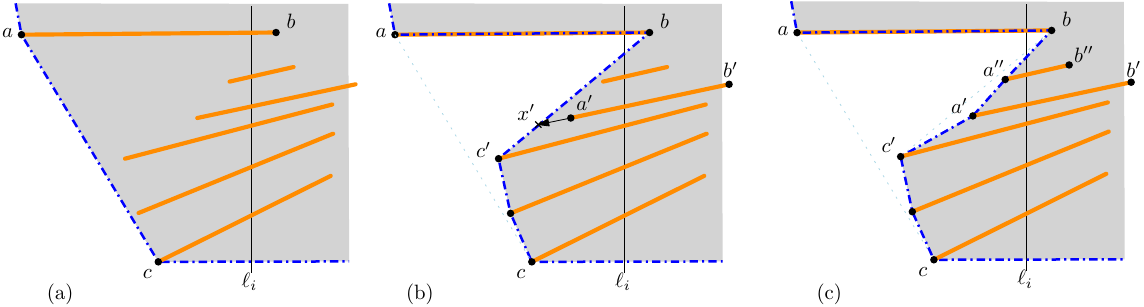}
	\caption{The left endpoint of segment $ab\in \widehat{S}_1$ is in the lower arc of $P$.
	(b) The result of operation \buildCap$(P,1,a)$.
	(c) The result of a subsequent operation \dip$(P,a',b')$.}
	\label{fig:ph1-correct}
\end{figure}

Second, consider a reflex vertex created by an operation ${\buildCap}(P,\varrho(a),a)$. Without loss of generality, assume that $a$ is in the lower arc of $P$, hence $\varrho(a)=1$. Refer to Fig.~\ref{fig:ph1-correct}. 
The operation creates an unsafe reflex vertex at $b$,
which is the right endpoint of $ab\in \widehat{S}_1$. 
Vertex $b$ is incident to two edges of the frame: $ab$ and $bc'$, where $c'$ is first vertex of $\carc(b,a,c)$ after $b$. Vertex $a$ cannot be an unsafe reflex vertex of $P$ by \ref{R1}, since it is the left endpoint of segment $ab\in \widehat{S}_1$. Segment $ab$ remains an edge of the lower arc $P$ in the remainder of Phase~1, since neither operation modifies such edges of $P$. If $c'$ is an interior vertex of $\carc(b,a,c)$, then it is a convex vertex of $P$, and remains convex. Assume that $c'=c$, and note that $c$ cannot be an unsafe reflex vertex or else it would have been created in a previous $\buildCap$ operation, which would imply that $bc\in \widehat{S}_1$, contradicting $ab\in \widehat{S}_1$. At the end of operation ${\buildCap}(P,\varrho(a),a)$, the polygon $\conv(a,b,c')$ is empty and lies in the exterior of $P$, so vertex $b$ satisfies \ref{R1} at that time. 

We show that if an unsafe reflex vertex $v_i$ of $P$ satisfies \ref{R1} in the course of Phase~1, it continues to satisfy \ref{R1} in the remainder of Phase~1. 
We may assume that $v_i=b$ is the right endpoint of $ab\in \widehat{S}_1$,  $v_{i-1}=a$, and $v_{i+1}=c'$. 
As noted above, $ab$ remains and edge of $P$ in Phase~1,
and neither $a$ nor $c'$ can become an unsafe reflex vertex. 
Since $c'$ is on the lower arc of $P$, we have $\varrho(c')=\varrho(a)$, and subsequent \buildCap\ operations cannot modify $bc'$. 
However, a subsequent \dip\ operation may modify $bc'$.
Suppose ${\dip}(P,a',b')$ is the first such operation, and ray $\overrightarrow{b'a'}$ hits $bc'$ at $x'$. 
This operation replaces $bc'$ by $\carc(b,x',a')$ and $\carc(a',x',c')$, where $a'$ is a safe reflex vertex, and all interior vertices of $\carc(b,x',a')$ and $\carc(a',x',c')$ are convex. Denote by $a''$ the neighbor of $b$ in $\carc(b,x',a')$. We need to show that $\conv(a,b,a'')$ lies in the exterior of $P$ after the operation, or equivalently, that $a''$ is to the right of the line through $ac'$. 

By condition (b) in Phase~1, we have $a'b'\in Q_i$. Therefore, $a'$ is on the left side of $\ell_i$. Since the segments in $Q_i$ are sorted in decreasing order by slope, both $a'$ and $x'$ are below the line through $ab$, consequently  $\carc(b,x',a')$ lies in the triange formed by the lines through $ab$,  $a'b'$, and $ax'$. The line through $a'b'$ crosses both $bc'$ and $ac'$, hence $a''$ is to the right of the $ac'$, as required. This completes the proof that \ref{R1} is maintained.
\end{proof}

The next lemma helps identify the segments in $Q_i$, $i\in \{1,\ldots, r\}$, whose left endpoints are \emph{not} in $P$.

\begin{lemma}\label{lem:middle}
Let $P$ be the frame returned by Phase~1. Let $i\in \{1,\ldots, r\}$, and let $Q_i=\{a_jb_j: j=1,\ldots ,|Q_i|\}$ be sorted in increasing order by the $y$-coordinates of $a_jb_j\cap \ell_i$.
If $a_j$ is a vertex in the lower (resp., upper) arc of $P$,
then so is $a_{j'}$ for all $j'<j$ (resp., $j'>j$).
\end{lemma}
\begin{proof}
Assume that $a_j$ is a vertex in the lower arc of $P$ (the argument is analogous when $a_j$ is in the upper arc). 
Suppose, for contradiction, that the lower arc of $P$ does not contain $a_{j'}$ for some $1\leq j'<j$. 
By \ref{F2}, $a_{j'}$ is in the interior of $P$, and so $\overrightarrow{b_{j'}a_{j'}}$ hits a segment in $\widehat{S}_1$ or an edge  of $P$. We claim that it hits an edge $uv$ of $P$ such that $uv\notin \widehat{S}_1$ and the left endpoint of $uv$ is an endpoint of some segment in $Q_i$. It follows that the algorithm would have applied ${\dip}(P,a_{j'},b_{j'})$ and $P$ would contain $a_{j'}$, contradicting the assumption.

To prove the claim, note that $Q_i$ is sorted by slope and so  $\overrightarrow{b_{j'}a_{j'}}$ does not cross any segment in $Q_i$. 
Since $j'<j$, point $a_j$ is above the line through $b_{j'}a_{j'}$. 
By \ref{F2}, the lower arc crosses $\ell_i$ at some point $p_i$ below $\ell_i\cap a_{j'}b_{j'}$. Let $\gamma$ be the portion of the lower arc of $P$ between $a_j$ and $p_j$. By \ref{F7}, the ray $\overrightarrow{b_{j'}a_{j'}}$ crosses $\gamma$. By \ref{F6}, $\gamma$ is in the right halfplane bounded by $L_{i-1}$. Since all segments in $\widehat{S}_1$ between $L_{i-1}$ and $L_i$ are in $Q_i$, the ray $\overrightarrow{b_{j'}a_{j'}}$ hits some edge $uv$ of $P$, where $uv\notin\widehat{S}_1$. The edge $uv$ can cross neither $a_jb_j$ nor $\gamma$ (it may be an edge of $\gamma$), therefore the left endpoint of $uv$ is to the left of $\ell_i$, hence it is an endpoint of a segment in $Q_i$.
\end{proof}

By Lemma~\ref{lem:middle}, the line segments in $\widehat{S}_1$ whose left endpoints are not in $P$ form a continuous interval. That is, for every $i\in \{1,\ldots, r\}$, there is a set of consecutive indices $M_i\subseteq \{1,\ldots , |Q_i|\}$ (possibly $M_i=\emptyset$ or $M_i=\{1,\ldots , |Q_i|\}$)
such that $j\in M_i$ if and only if the left endpoint of $a_jb_j$ is \emph{not} in $P$. Let $I=\{i\in \{1,\ldots , r\}: M_i\neq \emptyset\}$, $S_i'=\{a_jb_j: j\in M_i\}$ for all $i\in I$, and $S'=\bigcup_{i\in I} S_i'$.

\paragraph{Phase~2: Middle Segments.}
In Phase~2, we use the \shearDip\ and \dip\ operations to reach the left endpoints of \emph{at least quarter} of the segments in $S'$, followed by \buildCap\ operation to reach the right endpoints of those segments if necessary.
For all $i\in I$, let $a_{t(i)}$ be the leftmost left endpoint of a segment in the set $S_i'$, and denote this segment by $a_{t(i)}b_{t(i)}$.  Refer to Fig. 5.
We show (in Lemma~\ref{lem:xi} below) that a \shearDip$(P,a_{t(i)},x_i)$ operation can stretch the frame to reach $a_{t(i)}$ from some suitable point $x_i$ on an edge of $P$. We would like to choose all points $x_i$ on the same (upper or lower) arc of $P$. We decide which arc we use by comparing the number of segments in $S_i'$ above and below $a_{t(i)}b_{t(i)}$ for all $i\in I$. The set of segments above and below are $A$ and $B$, respectively, defined as follows:
\[A=\bigcup_{i=1}^r A_i, \mbox{ \rm where } A_i=\{a_jb_j: j\geq t(i),j\in M_i,\},\]
\[B=\bigcup_{i=1}^r B_i, \mbox{ \rm where } B_i=\{a_jb_j: j\leq t(i),j\in M_i,\}.\]
If $|A|\geq |B|$, then we reach the vertices $a_{t(i)}$ from the upper arc  for all $i\in I$; otherwise we reach them from the lower arc.  
Without loss of generality, assume that $|A|\geq |B|$.

The operations \shearDip$(P,a_{t(i)},x_i)$ may create unsafe reflex vertices at $a_{t(i)}$, $i\in I$. We need to ensure that $a_{t(i)}$ will not become an interior vertex of any $\carc$ created by subsequent \buildCap\ operations. 
We can protect a vertex $a_{t(i)}$, $i>1$, from a subsequent $\carc$ by ensuring that edges of the lower arc that cross the line $L_{i-1}$ do not change in phases 2-3 of the algorithm. 
We can guarantee this property for only some of the edges of the lower arc by setting the orientations $\varrho(v)$ for vertices of the frame that are right endpoints of segments in $\widehat{S}_1$. Let $f_1,\ldots ,f_h$ be the edges of the lower arc that cross any of the lines $L_1,\ldots , L_{r-1}$, ordered from left to right; note that $h\leq r-1$ as edge may cross several vertical lines. Now let
\[J=\{i\in I: L_{i-1} \mbox{ \rm crosses } f_j  \mbox{ \rm where } j \mbox{ \rm is odd}\}.\]
If $|\bigcup_{i\in J}A_i|\geq |\bigcup_{i\in I\setminus J}A_i|$, 
then we apply \shearDip$(P,a_{t(i)},x_i)$ for all $i\in J$; otherwise  for all $i\in I\setminus J$. Without loss of generality, assume that 
$|\bigcup_{i\in J}A_i|\geq |\bigcup_{i\in I\setminus J}A_i|$. 
We now assign orientations to the endpoints of the some edges $f_j$ as follows. If $v$ is the right (resp., left) endpoint of an edge $f_j$, where $j$ is odd, then let $\varrho(v)=-1$ (resp., $\varrho(v)=1$). Recall that the default orientation for all other vertices is $\varrho(v)=1$ if $v$ is on the lower arc of $P$, and $\varrho(v)=-1$ otherwise.

Before presenting Phase~2 of the algorithm, we need to specify the points $x_i$ for all $i\in J$. Consider the vertical upward ray from $a_{t(i)}$, and let $u_i$ be the first point on the ray that lies in the upper arc of $P$ or on a segment in $Q_i$.
If $u_i$ is in an edge of the upper arc, but not in a segment in $Q_i$, then let $x_i:=u_i$.
Otherwise, $u_i$ lies in some segment $ab\in Q_i$, which is either an edge or an internal diagonal of $P$.
Let $a_0a$ be the edge on the upper arc that precedes $a$ (in clockwise order).
The following lemma shows that we can choose a point $x_i$ in the interior of $a_0a$.
\begin{lemma}\label{lem:xi}
There exists a point $x_i$ in the interior of $a_0a$ such that $a_{t(i)}x_i$ does not cross $P$ or any segment in $\widehat{S}_1$.
\end{lemma}
\begin{proof}
We first claim that the interior of the triangle $\Delta(a_{t(i)}u_ia)$ is 
disjoint from segments in $\widehat{S}_1$. Suppose, to the contrary, 
that segment $a'b'\in \widehat{S}_1$, with left endpoint $a'$, 
intersects the interior of $\Delta(a_{t(i)}u_ia)$,
By construction, $a_{t(i)}b_{t(i)}, ab\in Q_i$, then $\Delta(a_{t(i)}u_ia)$ lies between $L_{i-1}$ and $L_i$, and $a'b'\in Q_i$. This means that $a'b'$ crosses $\ell_i$. 
The triangle $\Delta(a_{t(i)}u_ia)$ is strictly left of $\ell_i$. Therefore $a'b'$ must cross the boundary of $\Delta(a_{t(i)}u_ia)$. However $a'b'$ can cross neither $a_{t(i)}u_i$ (by construction) nor $au_i$ (as $au_i\subset ab$). Hence $a'b'$ must cross the edge $a_{t(i)}a$ of $\Delta(a_{t(i)}u_ia)$, and in particular the left endpoint $a'$ is in the interior of $\Delta(a_{t(i)}u_ia)$. Since $a_{t(i)}$ is the leftmost left endpoint of a segment in $Q_i$ that is not a vertex in $P$, we know that $a'$ is a vertex of $P$. 
By \ref{R1}, $a'$ is a convex or safe reflex vertex in $P$. 

Consider a maximal subchain of $P$ that intersects the interior of $\Delta(a_{t(i)}u_ia)$.
As the frame cross neither $a_{t(i)}u_i$ nor $au_i\subset ab$, both endpoints of such a subchain are on the side $a_{t(i)}a$ of $\Delta(a_{t(i)}u_ia)$. Consequently, its farthest vertex from $a_{t(i)}a$ is a reflex vertex of $P$. Without loss of generality, this reflex vertex is $a'$. Since $a'b'$ crosses the side $a_{t(i)}a$ of $\Delta(a_{t(i)}u_ia)$, it cannot subdivide the reflex angle of $P$ into two convex angles. 
Hence, $a'$ cannot be a convex or safe reflex vertex of $P$. This contradiction completes the proof of the claim.

It follows that the edges of $P$ are disjoint from the interior of $\Delta(a_{t(i)}u_ia)$.
In particular, edge $a_0a$ is disjoint from the interior of $\Delta(a_{t(i)}u_ia)$.
By \ref{R1}, $a$ is either a convex vertex or a safe reflex vertex of $P$.
In either case the angular domain $\angle a_0ab$ is convex,
hence it contains $a_{t(i)}a$. 
As no three segment endpoints are collinear,
there exists a point $x_i$ in the interior of $a_0a$ (in some neighborhood of $a$)
such that $a_{t(i)}x$ crosses neither $P$ nor any segment in $\widehat{S}_1$. 
\end{proof}

Phase~2 proceeds as follows:

\begin{enumerate}%\itemsep -2pt
\item For all $i\in J$, do:
    \begin{enumerate}
    \item[] set $P:={\shearDip}(P,a_{t(i)},x_i)$.
    \end{enumerate}
\item  While condition (a) or (b) below is applicable, do:
    \begin{enumerate}
    \item If there exists a segment $ab\in \widehat{S}_i$ such that the left endpoint $a$ is a vertex of $P$, but the right endpoint is not, then set $P:={\buildCap}(P,\varrho(a),a)$.
    \item Else if there exists a segment $ab\in Q_i$ for some $i\in \{1,\ldots, r\}$ such that $ab$ lies in the interior of $P$, $a$ is the left endpoint, and $\overrightarrow{ba}$ hits an edge $uv$ where $uv\not\in \widehat{S}_i$, and the left endpoint of $uv$ is an endpoint of some segment in $Q_i$, then set $P:={\dip}(P,a,b)$.
    \end{enumerate}
\item Return $P$, and terminate Phase~2.
\end{enumerate}

We show  that  Phase~2 maintains a frame. Instead of addressing \ref{R1}, as in Phase~1, we use properties~\ref{R2}--\ref{R3} below to establish \ref{F5}.

\begin{lemma}\label{lem:ph2correctness}
All operations in Phase~2 maintain \ref{F1}--\ref{F7} and the following properties.
\begin{enumerate}[resume*=R]
    \item\label{R2} every vertex $a_{t(i)}$, $i\in J$, has multiplicity 1 in $P$;
    \item\label{R3} if $v_i$ is an unsafe reflex vertex of $P$ and 
    $v_i\notin\{a_{t(j)}: j\in J\}$,  then 
    $v_i$ is the right endpoint of a segment in $\widehat{S}_1$, 
    vertices $v_{i-1}$ and $v_{i+1}$ are convex or have multiplicity 1 and the triangle $(v_{i+1},v_i,v_{i-1})$ is disjoint from the interior of $P$.
\end{enumerate}
\end{lemma}
\begin{proof}
The \shearDip\ operations can be performed independently. They each maintain \ref{F1}--\ref{F4} and \ref{F6}--\ref{F7} of the frame. We show that they also maintain \ref{F5}. The \shearDip\ operations introduce reflex vertices at $a_{t(i)}$ for all $i\in J$. 
Importantly, these operations do not modify the lower chain of $P$.
By Lemma~\ref{lem:ph1correctness}, $P$ has property~\ref{R1} at the beginning of Phase~2. Consequently, after performing all \shearDip\ operations in Phase~2, \ref{R1} holds for all unsafe reflex vertices with the possible exception of $a_{t(i)}$, $i\in J$, which have multiplicity 1. This implies \ref{R2}--\ref{R3}, hence \ref{F5}.

The \texttt{while} loop of Phase~2 is similar to Phase~1, but we need to prove that the multiplicity of every vertex $a_{t(i)}$, $i\in J$, remains 1. Specifically, we show that the changes incurred by \buildCap\ and \dip\ operations in the \texttt{while} loop are limited to regions that do not contain the vertices $a_{t(i)}$. For all $i\in J$, let $L_i^-$ be a polyline that consists of $x_ia_{t(i)}$ and a vertical line from $a_{t(i)}$ to the lower arc of $P$. 
Operation \shearDip$(P,a_{t(i)},x_i)$ creates two convex arcs incident to $a_{t(i)}$, on opposite sides of $L_i^-$: The interior vertices of the $\carc$ on the left of $L_i^-$ are to the left of $a_{t(i)}$, and by \ref{R1}, they cannot be left endpoints of segments in $Q_i$.
Therefore this $\carc$ is not involved in subsequent \buildCap\ or \dip\ operations for any segment in $Q_i$. That is, the impact of these operations is confined to the region to the right of $L_i^-$ and left of $L_j^-$, where $j>i$ is the next index in $J$ (if it exists). This further implies that the multiplicity of $a_{t(i)}$ remains 1 in the remainder of Phase~2, hence \ref{R2} is maintained. 

Analogously to the proof of Lemma~\ref{lem:ph1correctness}, all \buildCap\ and \dip\ operations in the \texttt{while} loop maintain \ref{F1}--\ref{F4} and \ref{F6}--\ref{F7}. They also maintain  \ref{R1} for unsafe reflex vertices other than $a_{t(i)}$, $i\in J$, which implies \ref{R2}. The combination of \ref{R2} and \ref{R3} yields \ref{F5}.
\end{proof}

\begin{figure}[htbp]
	\centering
	\includegraphics[width=.9\linewidth]{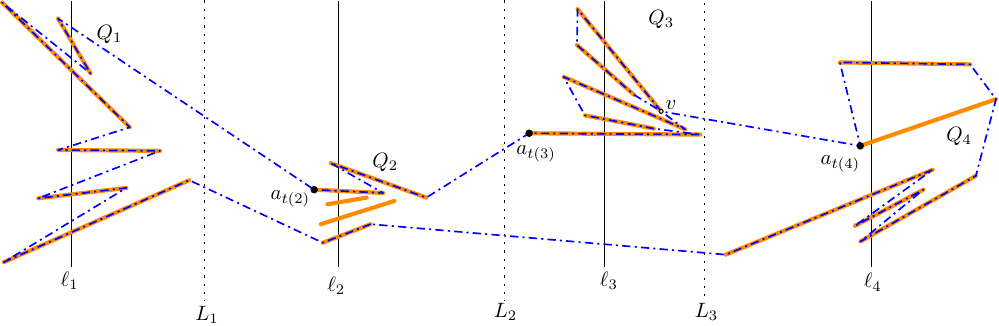}
	\caption{The set $\widehat{S}_1$ from Fig.~\ref{fig:circum2}, and frame $P$ at the end of Phase~2. Vertex $v$ has multiplicity~2.}
	\label{fig:circum3}
\end{figure}

\begin{lemma}\label{lem:A}
At the end of Phase~2, $P$ visits both endpoints of all segments in $\bigcup_{i\in J}A_i$.
Consequently, $P$ visits both endpoints of at least quarter of the segments in $\widehat{S}_1$.
\end{lemma}
\begin{proof}
The proof of the first claim is analogous to the proof of Lemma~\ref{lem:middle}. To prove the second claim, note that $a_{t(i)}$ is a vertex of the upper arc of $P$ for all $i\in J$. Consequently $a_{j'}$, for all $j'>t(i)$, is also part of the upper arc. The \buildCap\ operations in the while loop of Phase~2 ensure that if a left endpoint of a segment $a_{j'}b_{j'}\in Q_i$ is a vertex of $P$, then so is the right endpoint.
\end{proof}

\paragraph{Phase~3: Right Endpoints.}
At the end of Phase~2, $P$ visits both endpoints of every segment in $\bigcup_{i\in J}A_i$. However, it may visit only one endpoint of some segments in $\bigcup_{i\in I\setminus J}A_i$ and $B\setminus A$. In this phase, we use \buildCap\ operations to ensure that $P$ visits \emph{both} endpoints of these segments. Recall that we have already specified the orientation $\varrho(v)$ for every vertex $v$ of the frame $P$, depending on whether $v$ is the left or the right endpoint of a segment in $\widehat{S}_1$.
Phase~3 proceeds as follows.

\begin{enumerate}%\itemsep -2pt
\item While condition (a) below is applicable, do
    \begin{enumerate}
    \item If there exists a segment $ab\in \widehat{S}_i$ such that one endpoint, say $b$, is a vertex of $P$, but the other endpoint is not, then set
    $P:={\buildCap}(P,\varrho(b),b)$.
    \end{enumerate}
\item Return $P$, and terminate Phase~3.
\end{enumerate}

At the end of Phase~3, we obtain a frame $P$ that contains,  for each segment, either both endpoints or neither endpoint; see Fig.~\ref{fig:circum3}. Some vertices may have multiplicity 2, but the multiplicity of the special vertices $a_{t(i)}$, for all $i\in J$, remains 1.

\begin{lemma}\label{lem:ph3correctness}
All operations in Phase~3 maintain \ref{F1}--\ref{F5} and \ref{R2} for $P$.
\end{lemma}
\begin{proof}
Phase~3 applies only \buildCap\ operations, which always maintain \ref{F1}--\ref{F4}. 

We establish \ref{R2}. Similarly to Phase~2, the impact of any $\buildCap(P,\varrho(b),b)$ operation, where $b$ is a vertex of the upper arc, is confined to the region to the right of $L_i^-$ and left of $L_j^-$, where $j>i$ is the next index in $J$ (if it exists). 
We show that the same holds if $b$ is a vertex of the lower arc of $P$. 
Note that Phase~2 modifies only the upper arc of $P$, and so the lower arc at the beginning of Phase~3 is is the same as at the end of Phase~1.
Suppose, to the contrary, that Phase~3 modifies an edge that crosses $L_i^-$, for some $i\in J$. Let $\buildCap(P,\varrho(b),b)$ be the first such operation. Letting $c$ be the neighbor of $b$ in the frame in $\varrho(b)$ orientation, $\buildCap(P,\varrho(b),b)$ replaces edge $bc$ of $P$ with the polygonal path $ba+\carc(a,b,c)$. By the choice of $a_{t(i)}$, if a segment in $\widehat{S}_1$ crosses the vertical line through $a_{t(i)}$, it is in $Q_i$ and both of its endpoints are vertices of the frame at the end of Phase~1. Consequently, $b$ cannot be the left endpoint of a segment in $Q_i$, and so $bc$ crosses the vertical line $L_{i-1}$. At the beginning of Phase~3, $L_{i-1}$ crosses only one edge of the lower arc by \ref{F6}, say $f_{j'}$, 
where $j'$ is odd by the definition of $J$. However, the \buildCap\ operations in Phase~3 do not modify edge $f_{j'}$ due to carefully chosen orientation $\varrho(b)$. Indeed, if $b$ is to the left of $L_{i-1}$, then $\varrho(b)=-1$, otherwise $\varrho(b)=1$. In either case, both $b$ and $c$ are on the same side of $L_{i-1}$, and $bc$ does not cross $L_{i-1}$. This completes the proof that Phase~3 maintains \ref{R2}. 

It remains to handle \ref{F5}. Phases~1--3 use only operations \buildCap, \dip, and \shearDip, each of which creates at most one reflex vertex. The reflex vertices created by \dip\ operations are safe their multiplicities does not increase by Lemma~\ref{lem:safe}. The reflex vertices created by \shearDip\ operations have multiplicity 1 by \ref{R2}. Consider operation $\buildCap(P,\varrho(b),b)$, where $ab\in \widehat{S}_1$ and $c$ is the neighbor of $b$ in the frame in $\varrho(b)$ orientation. It creates a reflex vertex at $a$; the two neighbors of $a$ are $b$ and $c'$, where $c'$ is the first vertex of $\carc(a,b,c)$ adjacent to $a$ (possibly $c'=c$). Edge $ab$ remains an edge of $P$ until the end of Phase~3, since $ab\in \widehat{S}_1$, and our operations do not modify such edges. Vertex $b$ cannot be a reflex vertex created by a prior \buildCap\ operation. Assume that after operation $\buildCap(P,\varrho(b),b)$, vertex $a$ is part of a chain of two or more unsafe reflex vertices created by \buildCap\ operations. Then $\carc(a,b,c)=ac$ and $c$ is a reflex vertex created by a previous operation $\buildCap(P',\varrho(d),d)$, where $cd\in \widehat{S}_1$ and $\varrho(d)=-\varrho(b)$. As noted above, segment $cd$ remains an edge of the frame, and $d$ is not a reflex vertex created by a \buildCap\ operation. It follows that a chain $C$ of reflex vertices created by \buildCap\ operations has at most two vertices. If $C$ has only one vertex, then $C$ satisfies  \ref{F5} at that time, and possible \dip\ operations in Phases 1-2 maintain this property due to \ref{R3} in Lemma~\ref{lem:ph2correctness}. 
If $C$ has length 2, say $C=(a,c)=(v_i,v_{i+1})$, then one of the reflex vertices was added to chain in Phase 3, and subsequent \buildCap\ operations do not modify the chain $(b,a,c,d)=(v_{i-1},v_i,v_{i+1},v_{i+2})$ of the frame. Assume, without loss of generality, that the reflex vertex $v_i$ is created after $v_{i+1}$ in a \buildCap\ operation. Before the operation, the triangle $(v_{i+2},v_{i+1},v_{i-1})$ is disjoint from the interior of $P$. Then the quadrilateral $(v_{i+2},v_{i+1},v_i,v_{i-1})$ is also disjoint from the interior of $P$ after the operation, hence at the end of Phase~3. This confirms that Phase~3 maintains \ref{F5}.  
\end{proof}

\begin{lemma}\label{lem:bothendpoints}
Let $P$ be the frame at the end of Phase~3. If one endpoint of a segment in $\widehat{S}_1$ is a vertex in $P$, then so is the other endpoint.
\end{lemma}
\begin{proof}
The claim trivially holds when the \texttt{while} loop terminates.
\end{proof}

\paragraph{Phase~4: Obtaining a Simple Polygon.}
In the last phase of our algorithm, we set $P:={\chopWedges}(P)$. This is a valid operation by \ref{F5}. The resulting frame $P$ is a simple polygon whose vertex set is the same as at the end of Phase~3. By Lemma~\ref{lem:bothendpoints}, if one endpoint of a segment in $\widehat{S}_1$ is a vertex in $P$, then so is the other endpoint.
Consequently, $P$ is a circumscribing polygon for a set of segments in $\widehat{S}_1$,
which we denote by $\widehat{S}_2$. By Lemma~\ref{lem:A}, we have $|\widehat{S}_2|\geq |\widehat{S}_1|/4$. This completes the proof of Theorem~\ref{thm:circumscribe}.

\section{Disjoint Segments versus Disjoint Rays}
\label{sec:rays}

In this section, we give two sufficient conditions for an arrangement of disjoint segments to admit a circumscribing polygon. Both conditions involve extending the segments.

\begin{enumerate}[label={\rm (C\arabic*)},series=C]\itemsep 0pt
\item\label{C1} A set $S$ of $n$ disjoint line segments is \emph{extensible to rays} if there exists a set $R$ of $n$ disjoint rays, each of which contains a segment from $S$; see Fig.~\ref{fig:extensible}(left).
\item\label{C2} A set $S$ of $n$ disjoint line segments admits \emph{escape routes} if there exists an ordering and orientation of the segments $S=\{a_i,b_i:i=1,\ldots n\}$ such that if we shoot a ray from $b_i$ in direction $\overrightarrow{a_i b_i}$ for $i=1,\ldots ,n$, then each ray either goes to infinity or intersects a previous ray (without crossing any segment in $S$); see Fig.~\ref{fig:extensible}(right).
\end{enumerate}

\begin{figure}[htbp]
	\centering
	\includegraphics[width=.75\linewidth]{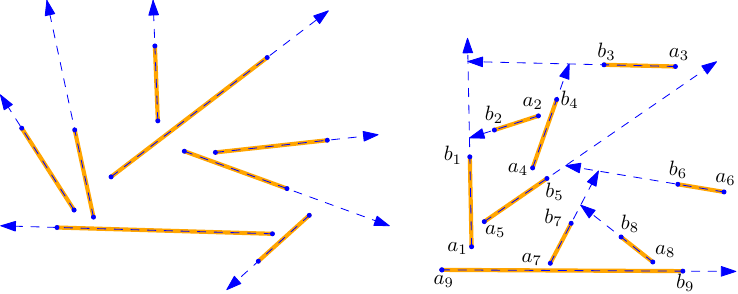}
	\caption{Left: an arrangement of disjoint segments extensible to rays.
Right: an arrangement of disjoint segments that is not extensible to rays, but admits escape routes.}
	\label{fig:extensible}
\end{figure}

Clearly, \ref{C1} implies \ref{C2}, but the converse is false in general. 
We can test \ref{C1} in $O(n\log n)$ time. Indeed, there are two possible extensions for a segment,  which can be encoded by a Boolean variable, and pairwise disjointness can be expressed by a 2SAT formula. 
However, we do not know whether condition \ref{C2} can be tested efficiently. 
Here we show that \ref{C1} and \ref{C2} each imply the existence of a circumscribing polygon.

\begin{theorem}\label{thm:escape}
If $S$ is a set of disjoint line segments satisfying \ref{C2},
then there is a circumscribing polygon for $S$.
\end{theorem}
\begin{proof}
By \ref{C2}, we may assume that $S=\{a_i b_i :i=1,\ldots n\}$ such that if we shoot a ray from $b_i$ in direction $\overrightarrow{a_i b_i}$ for $i=1,\ldots ,n$, then each ray either goes to infinity or intersects a previous ray. We call the part of the ray $\overrightarrow{a_i b_i}$ from $b_i$ to the first point where it intersects a previous ray or $\partial \conv(S)$ the \emph{extension} of segment $a_ib_i$.

Given the ordering and orientation of the segments in $S$, we construct a circumscribing polygon using the following algorithm, using the operations \buildCap\ and \dip\ introduced in Section~\ref{sec:circumscribe}; see Fig.~\ref{fig:escape} for an example.

\begin{enumerate}\itemsep -1pt
\item Initialize $P:=\partial \conv(S)$.
\item For $i=1$ to $n$:
    \begin{enumerate}
    \item[] If $b_i$ is not a vertex of $P$, then set $P:={\dip}(P,b_i,a_i)$.
    \end{enumerate}
\item For $i=1$ to $n$:
    \begin{enumerate}
    \item[] If $a_i$ is not a vertex of $P$, then set $P:={\buildCap}(P,1,b_i)$.
    \end{enumerate}
\item \chopWedges$(P)$.
\item Return $P$.
\end{enumerate}

We show that polygon $P$ maintains \ref{F1}--\ref{F5} from Definition~\ref{def:frame} (note that \ref{F6}--\ref{F7} are no longer applicable); and it also maintains the following invariant:
\begin{enumerate}[resume*=F]
\item\label{F8} After operation ${\dip}(P,b_i,a_i)$ (for $i=1,\ldots, n$), the extension of segment $a_ib_i$ lies in the exterior of $P$.
\end{enumerate}
Indeed, operations \buildCap\ and \dip\ always maintain \ref{F1}--\ref{F4}. For $i=1,\ldots ,n$, \ref{F8} follows from \ref{C2} and the definition of the \dip\ operation; and it guarantees the correctness of subsequent \dip\ operations.

\begin{figure}[htbp]
	\centering
	\includegraphics[width=.7\textwidth]{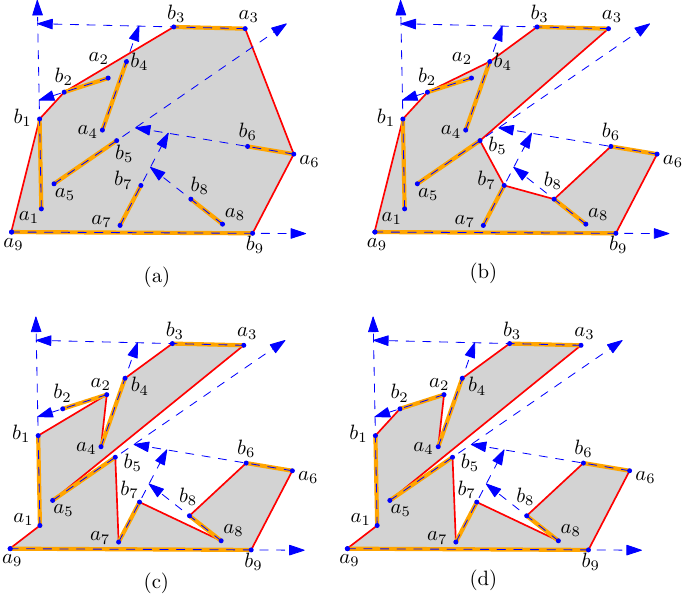}
	\caption{(a) An arrangement of 9 segments that admit escape routes from Fig.~\ref{fig:extensible}(right), and $P=\partial\conv(S)$.
	(b) Polygon $P$ after the \texttt{for} loop of \dip\ operations.
	(c) Polygon $P$ after the \texttt{for} loop of \buildCap\ operations.
	(d) The circumscribing polygon for $S$ after the \chopWedges\ operation.}
	\label{fig:escape}
\end{figure}

For \ref{F5}, note that each ${\dip}(P,b_i,a_i)$ operation creates a unique reflex vertex at $b_i$, and it is safe. At the end of the \texttt{for} loop of \dip\ operations, the vertices of $P$ are either convex vertices or safe reflex vertices. It remains to consider the \texttt{for} loop of \buildCap\ operations.

Each ${\buildCap}(P,1,b_i)$ operation creates a unique reflex vertex at $a_i$, which is unsafe and incident to both $a_ib_i$ and some edge $a_ic_i$. Since $a_ib_i\in S$, subsequent operations ${\buildCap}(P,1,b_j)$, $i<j$, modify neither $a_ib_i$ nor $a_ic_i$. Consequently, $b_i$ and $c_i$ are each convex or safe reflex vertices of $P$ before operation ${\buildCap}(P,1,b_i)$. As the operation can only decrease the interior angles at existing vertices,  $b_i$ and $c_i$  remain convex or safe reflex vertices in the remainder of the algorithm. This means that the maximum chain of unsafe reflex vertices of $P$ that that contain $a_i$ consists of a single vertex $a_i$. As  the triangle $\conv(a_i,b_i,c_i)$ is in the exterior of $P$ after the operation ${\buildCap}(P,1,b_i)$, \ref{F5} is maintained, as required. By \ref{F5},
the final \chopWedges\ operation is correct, and the algorithm returns a circumscribing polygon for $S$.
\end{proof}

The above results link the circumscribing polygon problem to the problem of extending line segments to rays. We now give an upper bound for the latter problem that seems to imply that the lower bound of the former problem (Theorem~\ref{thm:circumscribe}) is tight.

\begin{lemma}\label{lem:lowerbound}
For every $n\in \mathbb{N}$, there is a set $S$ of $n$ disjoint line segments in the plane such that the cardinality of every subset $S'\subseteq S$ that admits an escape route is $|S'|\leq 2\lceil\sqrt{n}\rceil-1$.
\end{lemma}
\begin{proof}
First, assume that $n=k^2$ for some $k\in \mathbb{N}$ (we will consider the other case later). Consider the circle $C$ of unit radius centered at the origin, and place $k$ points $a_1,\ldots , a_k$ in clockwise order along the arc of $C$ lying in the first quadrant (refer to Fig.~\ref{fig:lowerbound}). For every $i\in \{1,\ldots, k\}$, create a line segment $a_ib_i$ of length 4 tangent to $C$ such that $a_i$ is its left endpoint;
let $S_0=\{a_ib_i:i=1,\ldots k\}$. Finally, for every $i\in \{1,\ldots, k\}$, let $C_i$ be a circle of unit radius passing through $b_i$ and tangent to $a_ib_i$. We create $k-1$ additional segments that are almost parallel to $a_ib_i$. That is, the difference in slopes between any two of them is at most $\varepsilon$ (for some $\varepsilon$ small enough so that none of the segments cross). We make these segments all of equal length (4 units) and all tangent to  $C_i$. Let $S_i$ represent the $k$ segments tangent to $C_i$ ($i\in \{1,\ldots, k\}$).
This completes the construction for a set of segments $S=\bigcup_{i=1}^k S_i$. Notice that if we pick a sufficiently small $\varepsilon>0$ no two segments will cross.

\begin{figure}[htbp]
	\centering
	\includegraphics[width=.5\linewidth]{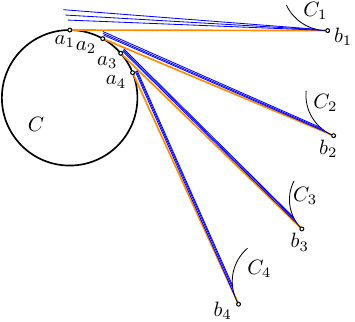}
	\caption{Our lower bound construction for $k=4$.}
     	\label{fig:lowerbound}
\end{figure}

Let $Q\subset S$ be a subset of segments that admit escape routes. We claim that $|Q|<2k$. Indeed, let $i_1<\ldots< i_k$ be the indices such that $Q\cap S_{i_j}\neq \emptyset$ (for all $j\leq k$).

First note that $Q\cap S_{i_j}$ can contain only one segment for all $j<k$: the left extension of each segment in $Q\cap S_{i_j}$ must hit any segment in $Q\cap S_{i_{j+1}}$,
thus all segments in $Q\cap S_{i_j}$ must be directed to the right. However, if $|Q\cap S_{i_j}|\geq 2$, then the right extensions of one of these segments hits another segment in the same set. Therefore $|Q\cap S_{i_j}|=1$ (for all $j\leq k$).

Overall, we have $k$ different groups $S_i$, each with $k$ segments. Any subset that admits an escape route contains at most one segment from all but one group and can potentially contain all segments of the last one. Since there are $k$ groups, the result follows.

To complete the proof it remains to consider the case in which $n$ is not a perfect square. The key trick is that our construction is hereditary (that is, even if we remove a few segments from our construction, the same upper bound holds). Thus, given $n\in\mathbb{Z}$, apply the construction described above for $k=\lceil\sqrt{n}\rceil$. This construction will have $\lceil\sqrt{n}\rceil ^2 \geq n$ segments. Remove segments arbitrarily until only $n$ remain, and apply the same argument to any subset that admits an escape route. This will give an upper bound of $2k -1=2\lceil\sqrt{n}\rceil-1$ as claimed.
\end{proof}

\section{Hardness for Circumscribing Polygons}
\label{sec:hardness2}

In this section we prove that it is NP-hard to decide whether a given set of disjoint line segments admits a circumscribing polygon. We start by proving NP-hardness for circumscribing a set of disjoint straight-line cycles (Section~\ref{ssec:HardCycles}), and then modify this reduction so that it works for disjoint line segment, as well (Section~\ref{ssec:HardSegments}). 

\subsection{Hardness for Disjoint Straight-Line Cycles}
\label{ssec:HardCycles}

Garey et al.~\cite[p.~713]{GareyJT76} proved that \textsc{Hamiltonian Path in 3-Connected Cubic Planar Graphs} (HP3CPG) is NP-complete. They reduce from 3SAT, and their reduction produces 3-connected planar graphs in which any Hamiltonian path connects two possible vertex pairs. This implies that the problem remains NP-complete if one endpoint of a Hamiltonian path is given. Since the graph is 3-regular, the problem remains NP-complete if the first (or last) edge of a Hamiltonian path is given. We call this problem
\textsc{Hamiltonian Path in 3-Connected Planar Cubic Graphs with Start Edge} (HP3CPG-SE): Given a 3-connected cubic planar graph $G=(V,E)$ and an edge $uv\in E$, decide whether $G$ has Hamiltonian path whose first edge is $uv$.
We reduce HP3CPG-SE to deciding whether a given PSLG admits a circumscribing polygon.
%, and then further reduce it to deciding whether a set of disjoint line segments admits a circumscribing polygon.

\begin{theorem}\label{thm:hardness2}
It is NP-complete to decide whether a given PSLG admits a circumscribing polygon, even if the PSLG is regular-degree-2.
\end{theorem}
\begin{proof}
Let $G=(V,E)$ and $uv\in E$ be an instance of HP3CPG-SE. We construct a PSLG $\widehat{G}=(\widehat{V},\widehat{E})$ in four steps below, and then show that $G$ has a Hamiltonian circuit containing $uv$ if and only if $\widehat{G}$ admits a circumscribing polygon.

\paragraph{Construction of a PSLG $\widehat{G}=(\widehat{V},\widehat{E})$.} 
Let $G=(V,E)$ be a 3-connected cubic planar
graph, $uv\in E$, and let $n=|V|$. 
We construct a PSLG $\widehat{G}=(\widehat{V},\widehat{E})$ as follows.

\paragraph{Step~1.} 
We first modify $G$ in the neighborhood of vertex $u$, and then specify a straight line embedding as follows. Refer to Fig.~\ref{fig:hardness2}. Subdivide each of the three edges incident to $u$ into a path of length 2, where $uv$ is subdivided into $(u,u',v)$. Apply a $Y\Delta$-transform at $u$, which creates a triangular face incident to $u'$, denoted $\Delta_u=\Delta(abu')$. Let $G_1$ be the resulting graph. Note that $G_1$ is a 3-connected cubic planar graph, and $G$ contains a Hamiltonian path starting with $uv$ if and only if $G_1$ contains a Hamiltonian path starting with $(a,b,u',v)$.

We construct an embedding of $G_1$ in which the outer face is $\Delta_u$. Using a result by Chambers et al.~\cite{ChambersEGL12}, $G_1$ admits a straight-line embedding in $\mathbb{R}^2$ such that (1) $\Delta_u$ is the outer face, (2) every interior face is strictly convex, (3) every vertex has integer coordinates, and (4) the coordinates of the vertices are bounded by a polynomial in $n$, which does not depend on $G_1$. Such embedding can be computed in polynomial time. (The method in~\cite{ChambersEGL12} improves Tutte's spring embedding~\cite{Tutte60}, which may require coordinates exponential in $n$.)

\begin{figure}[htbp]
	\centering
	\includegraphics[width=\textwidth]{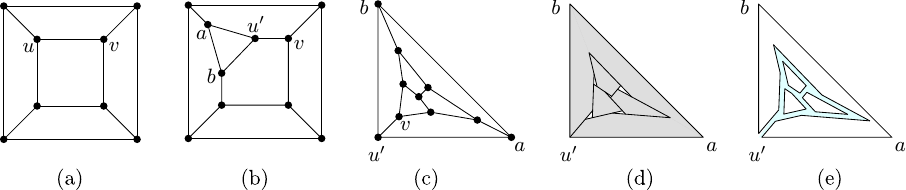}
	\caption{(a) A 3-connected cubic graph $G$ with  special edge $uv$.
    (b) Graph $G_1$ after a $Y\Delta$-transform around vertex $u$.
    (c) A convex embedding of $G_1$ such that the outer face is $\Delta_u$.
    (d)~Rotating the supporting line of internal edges.
    (e) Thickening the edges into corridors.
    }
	\label{fig:hardness2}
\end{figure}

\paragraph{Step~2.}
We call an edge of $G_1$ \emph{external} if it bounds the outer face, and \emph{internal} otherwise.
Let $\delta$ be the \emph{feature ratio} of $G_1$, i.e., the minimum distance between a vertex and a nonincident edge.
Delete the internal edges adjacent to $a$ and $b$; and denote by $G_2$ the resulting PSLG.

\paragraph{Step~3.} 
In this step, we split each internal vertex of degree-3 in $G_2$ into three vertices via a $Y\Delta$-transform, creating a new triangular face $\Delta (w_1w_2w_3)$.
Refer to Fig.~\ref{fig:junction}.
To do so, simultaneously rotate the supporting line of every internal edge about the midpoint of the edge in general, and about the endpoint $u'$ in case of edge $u'v$, by the same angle $\theta$.
Since $w$ has degree 3, the supporting lines of edges incident to $w$ will intersect at three different points: $w_1$, $w_2$, and $w_3$.
We choose the angle $\theta$ to be the maximum rotation so that $w_1$, $w_2$, and $w_3$ are still within $\frac{\delta}{2}$ distance from the original location of $w$ for all internal vertices $w$ in $G'$. Let $C_w$ be the circle obtained by dilating the inscribed circle of $\Delta (w_1w_2w_3)$ by $\frac{1}{2}$ from its center. Place a new edge $w_4w_5$ on the horizontal diameter of $C_w$. Denote by $G_3$ the resulting PSLG.

\paragraph{Step~4.} 
We define visibility for a PSLG as follows: Two points in the plane are \emph{visible} if the line segment between them is disjoint from the interior of opaque faces, where we consider the triangles obtained by $Y\Delta$-transforms to be \emph{transparent}, and all other faces \emph{opaque}. Opaque faces are shown gray in Figs.~\ref{fig:hardness2}(d) and \ref{fig:junction}(b).

Simultaneously shrink the opaque faces of $G_3$ using their straight skeleton by half of the minimum amount that would create a new visibility relation between a pair of vertices.
We assume that vertices $w_1$, $w_2$, and $w_3$ remain incident to the opaque faces in which they are strictly convex corners: in particular, each vertex belongs to a unique opaque face. Note that $w_4$ and $w_5$ are only visible from $w_1$, $w_2$, and $w_3$. We call the pair of parallel edges that used to be collinear a \emph{corridor}, and the triangle $\Delta (w_1w_2w_3)$ a \emph{chamber} (which is no longer a face).
Denote by $\widehat{G}=(\widehat{V},\widehat{E})$, the resulting PSLG. 
This concludes the construction of $\widehat{G}$.

\begin{figure}[htbp]
	\centering
	\includegraphics[width=.8\textwidth]{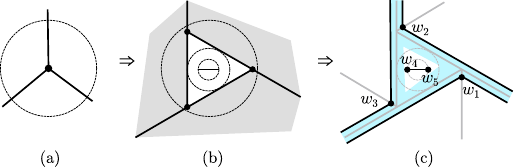}
	\caption{A chamber incident to three corridors and a segment in the chamber that sees only three other vertices. 
	}
	\label{fig:junction}
\end{figure}

\paragraph{Equivalence.} 
We claim that an HP3CPG-SE instance $G$ with edge $uv$ is positive if and only if $\widehat{G}$ admits a circumscribing polygon.
Assume $G$ admits a Hamiltonian path starting with edge $uv$. 
Then $G_1$ contains a Hamiltonian path $P_1$ that starts with $(a,b,u',v)$.
We obtain a circumscribing polygon $\widehat{P}$ of $\widehat{G}$ incrementally as follows. Initialize $\widehat{P}$ to be the convex hull of $\widehat{G}$.
Refer to Fig.~\ref{fig:circ-example}(a).
For each degree-2 vertex $w$ in $P_1$, add to $\widehat{P}$ both edges of each of the two corridors corresponding to the edges in $P_1$ incident to $w$.
Without loss of generality, assume that one edge from each of the two  
corridors is incident to vertex $w_3$. If $w$ corresponds to a chamber,
add the path $(w_1,w_5,w_4,w_2)$ to $\widehat{P}$ 
(which connect the other two edges of the corridors).
Refer to Figs.~\ref{fig:circ-example}(b--c).
For the last vertex $w$ in $P_1$, if $w$ corresponds to a chamber, w.l.o.g., let  $w_2$ and $w_3$ be the endpoints of the corridor that corresponds to the edge in $P$ adjacent to $w$. Add to $\widehat{P}$ the path $(w_3,w_4,w_5,w_2,w_1)$. Else, add the edge connecting the endpoints of such corridor.
By construction, every opaque face is in the interior of $\widehat{P}$.
Every edge of $\widehat{G}$ is either an edge of $P_1$ or a corridor in the interior of $\widehat{P}$. Hence, $\widehat{P}$ is a circumscribing polygon.

\begin{figure}[htbp]
	\centering
	\includegraphics[width=\textwidth]{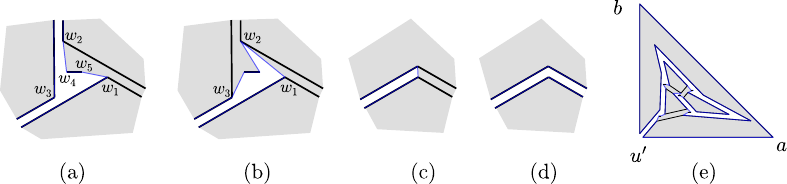}
	\caption{A circumscribing polygon $\widehat{P}$ at endpoints of corridors and a solution to the reduction in Fig.~\ref{fig:hardness2}. The shaded area is in the interior of $\widehat{P}$.}
	\label{fig:circ-example}
\end{figure}

Now assume that $\widehat{G}$ admits a circumscribing polygon $\widehat{P}$.
Note that the interior of an opaque face must be in the interior of $\widehat{P}$, or else some edge bounding the opaque face would be an external diagonal. Let $I$ be the (closed) convex hull of a corridor (recall that a corridor consists of two parallel edges of $\widehat{P}$). The corridors are constructed so that all visibility edges that intersect $I$ are induced by the four vertices of $I$.
Furthermore, the number of edges of $\widehat{P}$ in $I$ is exactly 0 or 2, or else one of the adjacent opaque faces would be in the exterior of $\widehat{P}$ (in case $\widehat{P}$ has 3 edges in $I$) or such edges would induce a cycle (in case $\widehat{P}$ has 4 edges in $I$).
If $I$ contains two edges of $\widehat{P}$, then either these are the edges of the corridor, or they are adjacent. If they are adjacent, one of them is a diagonal of $I$ and the other is an edge of the corridor; we call such an adjacent pair of edges a \emph{spike}.

We define the \emph{degree of a chamber} as the number of incident corridors in which  both edges are edges of $\widehat{P}$. We claim that the degree of every chamber is 1 or 2.  To prove the claim, consider a chamber in which the vertices are labelled $w_1,\ldots , w_5$ as in Fig.~\ref{fig:junction}(c). Notice that  $w_4$ and $w_5$ are visible from only three other vertices, and so $\widehat{P}$ must contain a path connecting two of the three vertices through the interior of the chamber. Assume, without loss of generality, that such a path is $(w_2,w_4,w_5,w_1)$ as in Fig.~\ref{fig:circ-example} (a). 
It follows that the degree of a chamber cannot be three, or else $w_1$ and $w_2$ would have degree 3 in $\widehat{P}$. Since the path $(w_2,w_4,w_5,w_1)$ is adjacent to the rest of the circuit $\widehat{P}$, it must have edges intersecting the convex hull $I$ of some corridor incident to the chamber.
If the degree of a chamber is zero, it can only be adjacent to spikes. 
Then $(w_2,w_4,w_5,w_1)$ and the adjacent spikes form a cycle that does not span all vertices of $\widehat{G}$, which is a contradiction. This completes the proof of the claim.

Consider a chamber of degree two, with vertices labeled as in Fig.~\ref{fig:junction}(c). 
Without loss of generality, $\widehat{P}$ contains the corridors incident to $w_1$ and $w_3$, and $w_2$ and $w_3$ as in Fig.~\ref{fig:circ-example}~(a).
Neither of the edges in the corridor incident to $w_2$ and $w_3$ can be in $\widehat{P}$ or else either $w_1$ or $w_2$ would have degree 3 in $\widehat{P}$.
Consequently, $\widehat{P}$ intersects each chamber of degree 2 exactly twice: once by the single-vertex path $(w_3)$ and once by $(w_2,w_4,w_5,w_1)$.
Then, such chamber cannot be adjacent to any spikes.
Now consider a chamber of degree one such that $\widehat{P}$ contains the corridor incident to $w_1$ and $w_3$ as in Fig.~\ref{fig:circ-example}~(b).
Then, $\widehat{P}$ contains a path from $w_2$ to $w_3$ containing the edge $w_4w_5$ (as in Fig.~\ref{fig:circ-example}~(b)), or a path from $w_1$ to $w_3$ containing the edge $w_4w_5$. 
In the first case, it is possible that the chamber is incident to a spike $(w_1,v,w_2)$ where $v$ is not a vertex in the chamber. 
In the latter, $w_2$ is not visited locally by the subset of $\widehat{P}$ intersecting the chamber and, therefore, must be visited by a spike adjacent to an adjacent chamber.
Since these are the only possible occurrences of spikes, we can modify any solution by shortcutting spikes of the form $(w_1,v,w_2)$ to edges $w_1w_2$ in the first case, and, in the latter case, modifying $\widehat{P}$ to traverse the chamber as in Fig.~\ref{fig:circ-example}~(b).
Hence, from now on we may assume that $\widehat{P}$ does not contain any spikes.

For a circumscribing polygon $\widehat{P}$ of $\widehat{G}$, we define a subgraph $P_1$ of $G_1$ as follows: If both edges of a corridor belong to $\widehat{P}$, then let the edge corresponding to the corridor be in $P_1$. Since the degree of each chamber in $\widehat{G}$ is 1 or 2, every interior vertex of $G_1$ has degree 1 or 2 in $P_1$, hence $P_1$ is a union of paths. A path that is not adjacent to the corridor incident to $u'$ will induce a cycle. Hence, there can be at most one path.
Since every chamber must be visited by $\widehat{P}$, 
this path visits every vertex except for $a$ and $b$.
Then, we can easily obtain a Hamiltonian path of $G_1$ starting with $(a,b,u',v)$; and a Hamiltonian path starting with edge $uv$ in the original graph $G$.
\end{proof}

\subsection{Hardness for Disjoint Line Segments}
\label{ssec:HardSegments}

In this section, we modify the constructions from Section~\ref{ssec:HardCycles} and extend our hardness result to the case of disjoint line segments. In the previous reduction, disjoint cycles are used to trace out the convex faces of an embedding of a graph $\widehat{G}$. To simulate a convex cycle with a geometric matching, we create \emph{turn gadgets} that simulate the corners of a convex polygon.

\begin{figure}[htbp]
	\centering
	\includegraphics[width=0.55\textwidth]{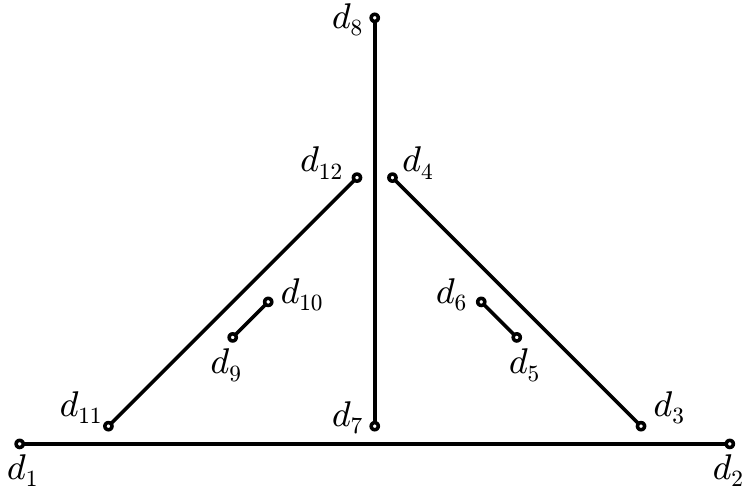}
	\caption{A diode}
	\label{fig:diode}
\end{figure}

Fig.~\ref{fig:diode} depicts an arrangement of 6 disjoint line segments, which we call a \emph{diode}, that, by itself, does not admit a circumscribing polygon~\cite{Grunbaum94}. The visibility relationships within a doide that guarantee this property are the following.
\begin{enumerate}[label={\rm (V\arabic*)},series=V]\itemsep 0pt
    \item \label{prop:1} 
    A point in $\{d_5,d_6\}$ ($\{d_9,d_{10}\}$) can see at most three points outside of this set, namely $d_3$, $d_4$, and $d_7$ ($d_7$, $d_{11}$, and $d_{12}$).
     \item \label{prop:2}
    Point $d_1$ ($d_2$) does not see $d_3$ ($d_{11}$).
\end{enumerate}
Since polygonizations are affine invariants, a nondegenerate affine image of a diode does not admit a polygonization, either. 

Using two affine copies of the diode as subconfigurations we can now construct a turn gadget. Fig.~\ref{fig:turn} shows this construction, but, for clarity, the positions of the points are not accurate. We first describe the important visibility properties in the gadget, and then show how to obtain coordinates for the points in order to achieve them.
We refer to the points in the left diode as $d_1,\ldots, d_{12}$ as in Fig.~\ref{fig:diode}, and to the points in the right diode as $d_1',\ldots, d_{12}'$.
We refer to points in $\{p_1,q_1,a,b\}$ as \emph{access points}, and to points in the two diodes, $c_1, c_2,c_3$, or $c_4$ as \emph{interior points}.
We construct the gadget so that 
\begin{enumerate}[resume*=V] \itemsep 0pt
\setcounter{enumi}{2}
    \item \label{prop:3} 
    interior points can only see other interior points and access points;
    \item \label{prop:4} 
    $d_2$ ($d_2'$) sees $p_1$ ($q_1$), but not $q_1$ ($p_1$); 
    \item \label{prop:5} 
    $c_3$ and $c_4$ can only see $c_1$, $c_2$, $d_2$, and $d_2'$; 
    \item \label{prop:6} 
    a point in $\{d_3,\ldots,d_{12}\}$ ($\{d_3',\ldots,d_{12}'\}$) can see at most three points outside of this set, namely $d_1$, $d_2$, and $a$ ($d_1'$, $d_2'$, and $b$); 
    \item \label{prop:7} 
    a point in $\{d_9,\ldots,d_{12}\}$ ($\{d_9',\ldots,d_{12}'\}$) can see at most three points outside of this set, namely $d_1$, $d_7$, and $d_8$ ($d_1'$, $d_7'$, and $d_8'$).
\end{enumerate}

\begin{figure}[htbp]
	\centering
	\includegraphics[width=0.95\textwidth]{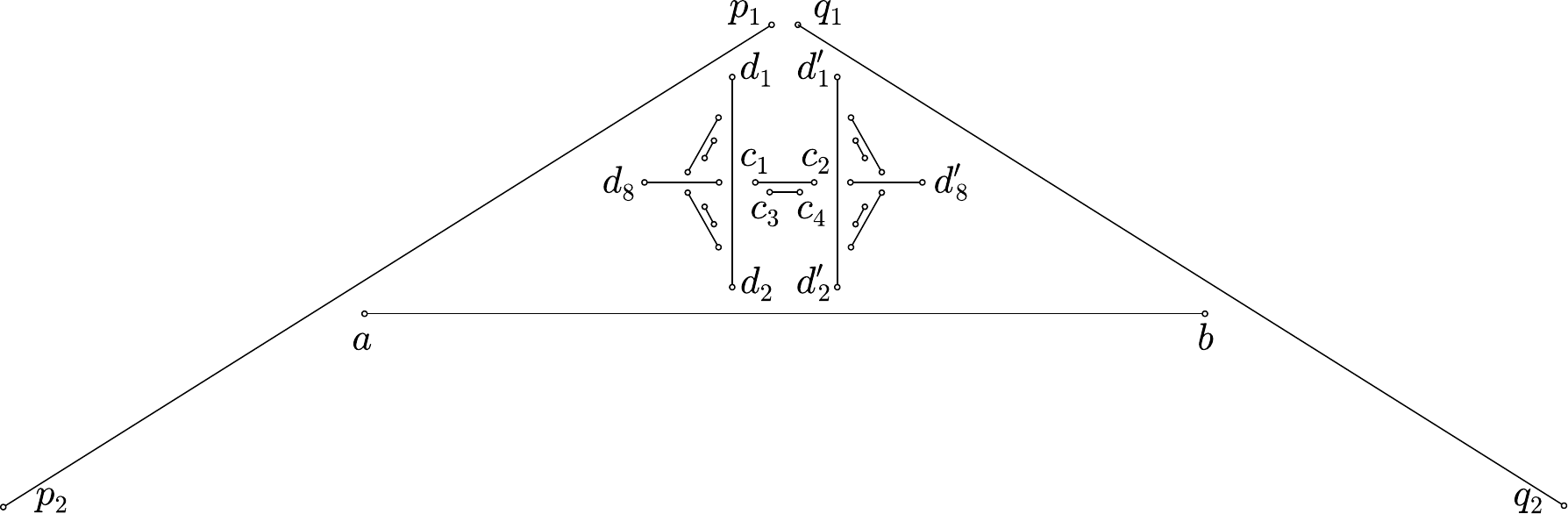}
	\caption{Turn gadget}
	\label{fig:turn}
\end{figure}

\paragraph{Construction of the Turn Gadget.}
Each turn gadget is constructed by replacing a vertex $w_i$ of a chamber $\Delta(w_1 w_2 w_3)$ as shown in Fig.~\ref{fig:turn-construction}.
We simultaneously perform the steps below for every chamber vertex $w_i$.
We describe the construction for $w_1$; refer to Fig.~\ref{fig:turn-construction}(a).
Let $\eps$ be one tenth of the feature ratio of the graph $\widehat{G}$.
Let $a'$ and $b'$ be points on the edges incident to $w_1$ at distance $\eps$ from $w_1$.
If we consider $w_1$ transparent, by construction, the only vertex that can (partially) see the left side of segment $a'b'$ is $w_2$. 
Let $e$ be the intersection of the supporting line of $w_1w_2$ and $a'b'$. 
By construction, $e$ must lie in the interior of the segment from the midpoint of $a'b'$ to $b'$.
We shrink the two edges incident to $w_1$ creating new endpoints $p_1$ and $q_1$ at distance $\delta>0$ from $w_1$.
We choose $\delta>0$ to be the maximum value so that, for every turn gadget, 
if any new endpoint (at distance $\delta$ from some vertex in $\widehat{V}$) 
outside of a turn gadget sees a point $c\in ab$, then $|ce|\leq \frac13\min\{|ea'|,|eb'|\}$ 
(in Fig.~\ref{fig:turn-construction}(b), $\delta$ is the radius of the dotted circles).
At each turn gadget, the upper bound for $\delta$ can be expressed as a bounded-degree rational function of the coordinates of the gadget, therefore the minimum over all turn gadgets is also bounded-degree rational function over the coordinates of the vertices in $\widehat{V}$.

\begin{figure}[htbp]
	\centering
	\includegraphics[width=\textwidth]{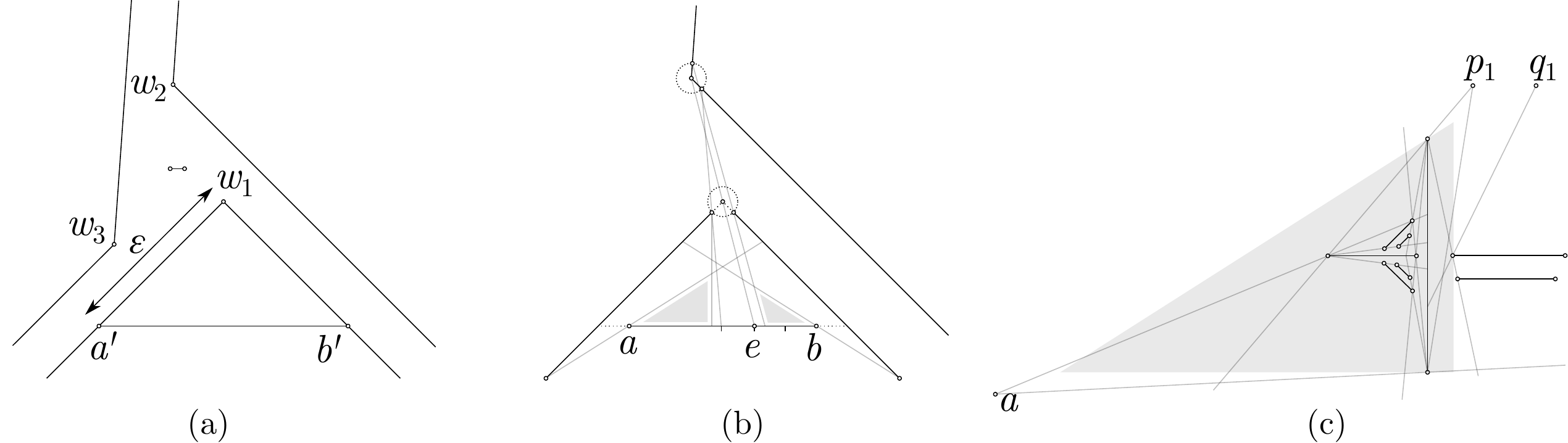}
	\caption{Turn gadget construction. For clarity, $\eps$ is not shown in proportion.}
	\label{fig:turn-construction}
\end{figure}

Assume, without loss of generality, that $a'b'$ is horizontal and $w_1$ lies above its supporting line. 
Add a segment $ab\subset a'b'$ by shrinking $a'b'$ so that $a$ ($b$) is at distance at most $\frac13\min\{|ea'|,|eb'|\}$ from $a'$ ($b'$).
By construction, there is a triangle $T_1$ ($T_2$) on top of $ab$ whose interior can only be seen by $p_1$, $q_1$, $a$, and $b$.
Let $T_1'$ be the maximum triangle in $T_1$ so that it has a vertical side, i.e., whose supporting line is perpendicular to $ab$, so that $p_1$ is on or to the right of such supporting line; and define $T_2'$ analogously. As long as we construct the diodes in the interior of $T_1'$ and $T_2'$, \ref{prop:3} is satisfied and all remaining necessary visibility constraints (constraints \ref{prop:1}--\ref{prop:2} and \ref{prop:4}--\ref{prop:7} become local, i.e., the positions of the points depend only on other points of the gadget). 
Since the constraints translate to linear inequalities, where each line is spanned by two vertices, we can compute the coordinates of all the interior points of the turn gadget in polynomial time. Fig.~\ref{fig:turn-construction}(c) illustrates linear constraints by gray lines.

\paragraph{Properties of the Turn Gadget.} 
The following four lemmas establish the key properties of a turn gadget. We show that if a set of disjoint line segments $S$ contains a turn gadget, then in every circumscribing polygon for $S$, there is a polygonal chain between $p_1$ and $q_1$ that visits all vertices of the gadget, and no other vertices (Lemma~\ref{lem:gadget-traversal}). We begin with a lemma that handles only one diode in the turn gadget.

\begin{lemma}\label{lem:diode-path}
Let $S$ be a set of disjoint line segments that contains a turn gadget with properties \ref{prop:1}--\ref{prop:7}.
Every circumscribing polygon $P$ for $S$ contains a path $\mathcal{P}$ on its boundary with one endpoint at $a$ and the other at $d_1$ or $d_2$ and whose interior vertices are exactly $\{d_3,\ldots,d_{12}\}$.
\end{lemma}

\begin{proof}
Refer to Fig.~\ref{fig:diode-path} for examples of $\mathcal{P}$.
Let $\mathcal{B}$ be a the minimal path on the boundary of $P$ containing $\{d_3,\ldots,d_{12}\}$ such that its endpoints are in $\{d_1,d_2,a\}$.
By \ref{prop:6}, $\mathcal{B}$ exists and must contain only points in the left diode and $a$.
Further, the two endpoints of $\mathcal{B}$ cannot be $d_1$ and $d_2$ or else one could obtain a circumscribing polygon of a diode by removing $a$ (if present) from $\mathcal{B}$, since 
the triangle $\Delta(a d_i d_j)$ formed by $a$ and any pair of vertices in the diode visible from $a$ is empty. 
However, as noted above, a diode does not admit a circumscribing polygon~\cite{Grunbaum94}.
It follows that one endpoint of $\mathcal{B}$ is $a$.
Let $v$ be the vertex in $\mathcal{B}$ adjacent to $a$. 
Note that $v\notin\{d_1,d_2\}$, else we could omit $a$ from $\mathcal{B}$,
contradicting the minimality of $\mathcal{B}$.
Then $\mathcal{B}=\mathcal{P}$ as required.
\end{proof}

\begin{figure}[htbp]
    \centering
    \includegraphics[width=\textwidth]{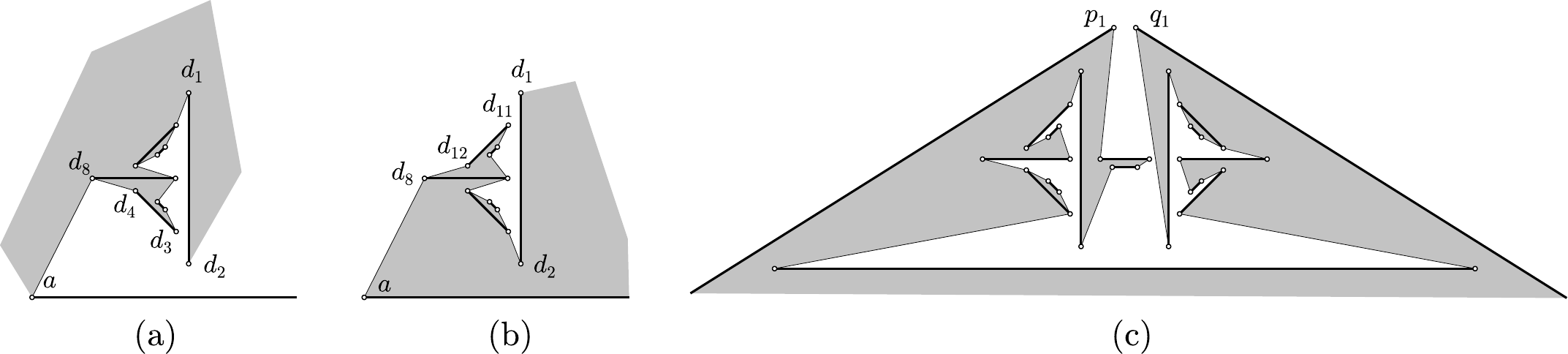}
    \caption{Examples of $\mathcal{P}$. The interior (exterior) of $P$ is shown in gray (white).}
    \label{fig:diode-path}
\end{figure}

Recall that every segment in $S$ is either an edge or an internal diagonal of a circumscribing polygon $P$. The union of $P$ and $S$ forms a plane graph that subdivide the plane into faces. 
Clearly, each face is either in the interior or the exterior of $P$. We define the \emph{left side} (resp., \emph{right side}) of a directed segment $ab\in S$ as the face of this graph on the left (resp., right) of $ab$. 

\begin{lemma}\label{lem:exterior}
Let $S$ be a set of disjoint line segments that contains a turn gadget satisfying \ref{prop:1}--\ref{prop:7}.
In every circumscribing polygon $P$ for $S$, the left side of the segment $d_2d_1$ is in the exterior of $P$, hence $d_2d_1$ is an edge of $P$.
\end{lemma}

\begin{proof}
Let $\mathcal{P}$ be the path described in Lemma~\ref{lem:diode-path} from $a$ to $d_1$ or $d_2$, in which $d_3,\ldots,d_{12}$ are interior vertices.
Suppose, for contradiction, that the left side of $d_2d_1$ is in the interior of $P$.
The left side of $d_2d_1$ contains $\Delta(d_1 d_2 d_7)$, as no visibility edge crosses this triangle by \ref{prop:2}. 
Both edges of $\mathcal{P}$ incident to $d_7$ are in the same closed halfplane bounded by the line through $d_7d_8$, or else $d_7d_8$ would be an external diagonal of $P$. 
Assume first that both edges are in the closed halfplane on the right of $d_7d_8$.
In particular, $d_7$ is adjacent to neither $d_5$ nor $d_6$.
By \ref{prop:1}, $d_5$ and $d_6$ must be adjacent to $d_3$ and $d_4$, respectively,
and $d_5d_6$ is an edge of $P$. In particular $f=(d_3,d_4,d_6,d_7)$ is a face in $P\cup S$. Furthermore, since no edge in $\mathcal{P}$ separates $\{d_5,d_6\}$ from $\Delta(d_1 d_2 d_7)$, the left side of $d_5d_6$ is in the interior of $P$. This implies that the $f$ is in the exterior of $P$. Now $d_3d_4$ must be an edge of $P$, or it would be an external diagonal,
and so $P$ contains all edges of the face $f$, which is a contradiction.
When both edges incident to $d_7$ are in the closed halfplnae on the left of $d_7d_8$, 
the argument is analogous by the symmetry of the diode. 
%
%Therefore $d_7d_8$ must be an edge of $\mathcal{P}$ or else it is an external diagonal of $P$. \csaba{I don't understand why $d_7d_8$ must be an edge of $\mathcal{P}$.} The first edge of $\mathcal{P}$ must connect $a$ to $d_3$ or $d_4$ or else the second edge of $\mathcal{P}$ is $d_8d_7$ and $\mathcal{P}$ cannot traverse the whole diode. Then $\mathcal{P}$ must traverse $\{d_3,\ldots,d_6\}$, then the edge $d_7d_8$ followed by $\{d_9,\ldots,d_{12}\}$ and finally $d_1$. In order to traverse $\{d_3,\ldots,d_6\}$, then the edge $d_7d_8$ without making $d_3d_4$ an external diagonal, $\mathcal{P}$ must visit $d_7$ before $d_8$. Then, the only planar way to traverse $d_8$, then $\{d_9,\ldots,d_{12}\}$ followed by $d_1$ makes $d_{11}d_{12}$ an external diagonal. This contradiction concludes the proof.
\end{proof}

\begin{lemma}\label{lem:ab}
Let $S$ be a set of disjoint line segments that contains a turn gadget satisfying \ref{prop:1}--\ref{prop:7}.
In every circumscribing polygon $P$ for $S$, the left side of $ab$ is in the exterior of $P$, hence $ab$ is an edge of $P$.
\end{lemma}

\begin{proof}
Suppose, for contradiction, that the left side of $ab$ is in the interior of $P$. 
Let $\mathcal{P}$ be the path described in Lemma~\ref{lem:diode-path} from $a$ to $d_1$ or $d_2$. Note that both $d_1d_2$ and $d_1'd_2'$ are edges of $P$ by Lemma~\ref{lem:exterior}.

Let $v$ and $v'$ be the two neighbors of $a$ in $P$ that are in $\mathcal{P}$ and not in $\mathcal{P}$, respectively.  Then $v'\neq d_2$, otherwise $\mathcal{P}\cup (v,d_2,d_1)$ would be a closed polygon in $P$; and $v'\neq d_2'$ otherwise the path $(v,d_2',d_1')$ is part of $P$ and the left side of $ab$ is the same as the right side of $d_2'd_1'$, contradicting Lemma~\ref{lem:exterior}.
It follows that the left side if $ab$ and the right side of $av$ are the same face of $P\cup S$, and so by assumption the right side of $av$ is in the interior of $P$ (cf.~Fig.~\ref{fig:diode-path} middle, where $v=d_8$).

We claim that the other endpoint of $\mathcal{P}$ is $d_1$ (and not  $d_2$). Suppose, to the contrary, that $\mathcal{P}$ is a path from $a$ to $d_2$. The last edge of $\mathcal{P}$ is $d_1d_2$ by Lemma~\ref{lem:exterior}. As $\mathcal{P}$ traverses both $av$ and $d_1d_2$ in the same direction, the interior of $P$ is on the right side of both. This means that the left side of $d_2d_1$ is in the interior of $P$, which contradicts  Lemma~\ref{lem:exterior}.

By a symmetric argument for the right diode, we conclude that there must be a subpath $\mathcal{P}'$ on the boundary of $P$ from $b$ to $d_1'$ going only through vertices of the right diode. By \ref{prop:5}, $c_3$ and $c_4$ can only be adjacent to $c_1$ and $c_2$. If the left side of $ab$ is in the interior of $P$, then $c_1c_2$ is an external diagonal. This contradicts our initial assumption, and concludes the proof.
\end{proof}

% \begin{lemma}\label{lem:no-edge}
% No edge of a circumscribing polygon $P$ containing a copy of the turn gadget intersects the interior of $\triangle p_1p_2a$ ($\triangle q_1q_2b$).
% \end{lemma}

% \begin{proof}
% We prove the claim for $\triangle p_1p_2a$ since the argument for $\triangle q_1q_2b$ can be obtained by symmetry.
% Suppose the contrary for contradiction.
% Because interior points cannot see $p_2$ by (1), such an edge must either have an endpoint in $p_1$ or block the visibility between $p_1$ and interior points. 
% In either case, $p_1$ cannot be connected to an interior point. 
% Since there are only three available access points, all interior points must be visited by $P$ contiguously, i.e., all interior points must be traversed by a path starting and ending on a remaining access point). 
% First, consider the case where $q_1$ and $a$ (resp., $b$) are such endpoints. Then such path must circumscribe the right (resp., left) diode, a contradiction. 
% Now, consider that the endpoints of the path are $a$ and $b$. 
% Then, $ab$ must be an internal diagonal of $P$.

% and each diode is traversed as option 2 in fig strategy.svg. But the new two central segments prevent this traversal, concluding the contradiction.
% \end{proof}

% \begin{lemma}\label{lem:interior-face}
% The triangle $\triangle p_1p_2a$ ($\triangle q_1q_2b$) is contained in any circumscribing polygon $P$ containing a copy of the turn gadget.
% \end{lemma}

% \begin{proof}
% TODO...
% \end{proof}

\begin{lemma}\label{lem:gadget-traversal}
Let $S$ be a set of disjoint line segments that contains a turn gadget satisfying \ref{prop:1}--\ref{prop:7}.
Every circumscribing polygon $P$ for $S$ must contain a chain between $p_1$ and $q_1$ whose interior vertices are exactly the vertices in the interior of the gadget and the edge $ab$;
furthermore, the right side of $p_2p_1$ and $q_1q_2$ are in the interior of $P$.
\end{lemma}

\begin{proof}
By \ref{prop:3}, the interior points of a turn gadget can see only access points, i.e., $\{p_1,q_1,a,b\}$. By Lemma~\ref{lem:diode-path}, $a$ and $b$ are each adjacent to some  interior point and, by Lemma~\ref{lem:ab}, $ab$ is an edge of $P$.
It follows that $P$ contains the required chain. 
See Fig.~\ref{fig:diode-path} (right) for an example.
Also by Lemma~\ref{lem:ab}, the right side of every edge in the directed path from $p_1$ to $q_1$ is in the interior of $P$.
If the right side of $p_2p_1$ ($q_1q_2$) is in the exterior of $P$, then $p_2p_1$ ($q_1q_2$) is an external diagonal, a contradiction.
% As shown in the proof of the previous lemma, a has at least one neighbor $v$ in $P$ which is in the interior of the gadget. Let $a'$ be the closest point to $a$ on the segment $p_1p_2$. By the previous lemma we know $aa'$ lies on the interior of $P$, and by definition of a polygonization $ab$ lies on the interior or boundary of $P$. Now let $u$ be the other neighbor of $a$ in $P$. Say $u$ is not in the turn gadget; so $u$ lies on the opposite side of $a'ab$ from $v$. So then the path $uav$ which is on the boundary of $P$ separates $a'$ from $b$, but this is a contradiction since we know the path $a'ab$ is fully contained in $P$. \oliver{I use ``contained in $P$'' assuming we view $P$ as a closed body; not sure if we should state that somewhere explicitly}. So both neighbors of $a$ (and by symmetry $b$) must be in the gadget. 
%
% By construction there are only 4 vertices (in the entire input) that can see the interior of the gadget: $a$, $b$, $p_2$, and $q_2$. Further, we know $a$ and $b$ have both their neighbors in the body of the gadget. So, if we take $S$ to be any maximal subpath of $P$ containing only vertices in the body of the gadget, we know that the first and last point of $S$ must be $p_2$ and $q_2$ (they cannot be equal). Since every vertex has degree 2 in a polygonization, and each of $p_2$ and $q_2$ have a neighbor outside the gadget, $S$ must be unique and thus contains every vertex in the body of the gadget \oliver{this last sentence needs rephrasing}.
\end{proof}

\paragraph{Construction of the Reduction.}
Let $\widehat{G}$ be a PSLG constructed in the proof of Theorem~\ref{thm:hardness2} (Section~\ref{ssec:HardCycles}). Recall that $\widehat{G}$ consists of line segments and opaque faces which are the interior of circuits. There are exactly two reflex vertices in the outer opaque face (the one that intersects the boundary of the convex hull of $\widehat{G}$), and all other opaque faces are convex.
We will replace the convex angles of all opaque faces as described in the construction of the turn gadget with two differences at a few special vertices, that we describe next.

\begin{itemize}
\item For vertices on the boundary of a chamber, the previous description is unchanged. 
\item For the four vertices in the convex hull (corresponding to vertices $a$, $b$, and $u'$ in Fig.~\ref{fig:disjoint-seg-full}(d)), there is no vertex $w_2$, i.e., no point that initially sees the left side of $ab$. This only decreases the number of constraints to construct the turn gadget.  
\item For the two reflex vertices of the outer opaque polygon, we adjust the turn gadget as follows. Refer to Fig.~\ref{fig:disjoint-seg-full}(a).
Let $w_1$ and $r_1$ be the convex and reflex vertices of $\widehat{G}$ corresponding to the vertex adjacent to $b$ in the interior of the convex hull of $G_1$.
Shrink one of the edges incident to $r_1$ by $\delta>0$, creating a new vertex $r_1'$. 
We proceed analogously for the reflex vertex $r_2$ corresponding to the vertex adjacent to $a$ in $G_1$.
The choice of which edge to shrink is made so that vertices $r_1'$ and $r_2'$ cannot see each other. %\csaba{Why is such a choice possible?}
We follow by replacing $w_1$ by a turn gadget as described before with $e$ defined by the point to the left of $ab$ visible from $r_1$ ($r_2$) by making $w_1$ transparent.
We require an additional constraint making $\delta$ small enough so that no new visibility is created (in particular, $r_1'$ ($r_2'$) does not see the segment $w_4w_5$ in the center of adjacent chambers).
\end{itemize}

\begin{figure}[htbp]
    \centering
    \includegraphics[width=\textwidth]{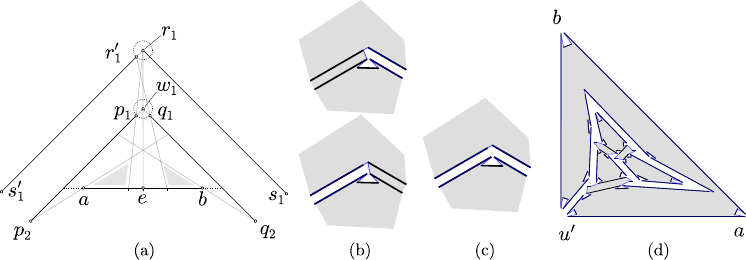}
    \caption{Reduction from the instance shown in Fig.~\ref{fig:hardness2} to disjoint segments. For clarity, only segment $ab$ from each turn gadget is shown.}
    \label{fig:disjoint-seg-full}
\end{figure}

\begin{lemma}\label{lem:equivalence-circ-disjoint}
The set of disjoint line segments $S$ constructed in the above reduction admits a circumscribing polygon if and only if the original instance of HP3CPG-SE is positive.
\end{lemma}

\begin{proof}
Assume the instance of HP3CPG-SE is positive. As described in the proof of Theorem~\ref{thm:hardness2}, we can construct a circumscribing polygon for the set of cycles created in the instance produced in Section~\ref{ssec:HardCycles}.
Given a polygonization of these cycles, Figures~\ref{fig:diode-path}~(c) and \ref{fig:disjoint-seg-full}(b--d) show how to transform this into a circumscribing polygon of $S$. For every turn gadget, replace each $w_1$ by its corresponding path from $p_1$ to $q_1$ shown in Fig.~\ref{fig:diode-path}(c).
Figures~\ref{fig:disjoint-seg-full}(b--c) show how to replace paths such as in Fig.~\ref{fig:circ-example}(c--d) that traverse the two reflex vertices.

Conversely, assume that $S$ admits a circumscribing polygon $P$, with edge set $E(P)$. 
By construction, $S$ is composed of the boundaries of convex opaque polygons whose corners have been replaced with turn gadgets, one opaque polygon with two reflex angles modified as shown in Fig.~\ref{fig:disjoint-seg-full}(a), and a segment $w_4w_5$ in the center each chamber (as seen in Fig.~\ref{fig:turn-construction}(a)). The segments $w_4w_5$ in the chambers are unchanged from the instance described in Section~\ref{ssec:HardCycles}, and due to Lemma~\ref{lem:gadget-traversal}, in any circumscribing polygon of $S$ the turn gadgets must behave just as the corners of polygons in the reduction used to prove Theorem~\ref{thm:hardness2}, i.e., we can consider the points in the gadget as a single point. Thus, it remains only to show that the two reflex angle gadgets also behave as the corners of a polygon in circumscribing polygon of $S$.

We first claim that if $r_1r_2\notin S$, then $r_1r_2\notin E(P)$. 
Refer to Fig.~\ref{fig:disjoint-seg-full}.
Suppose, for contradiction, that $r_1r_2\notin S$ and $r_1r_2\in E(P)$. 
Note that $r_1$ and $r_2$ are the only reflex vertices of an opaque polygon in the graph $\widehat{G}$. Since $r_1r_2\notin S$, the edge $r_1r_2$ partitions this opaque face into two faces, each of which has at least one additional (convex) vertex. The convex vertices of $\widehat{G}$ have been replaced by turn gadgets in the instance $S$. By Lemma~\ref{lem:gadget-traversal}, all vertices in turn gadgets are in subpaths of $P$ whose endpoints are visible to neither $r_1$ nor $r_2$.
By Lemma~\ref{lem:ab}, both sides of $r_1r_2$ are in the interior of $P$,
which is a contradiction. This proves the claim.

By construction, $r_1$ and $r_2$ see each other, but $r_1'$ and $r_2'$ do not. 
Further, by Lemma~\ref{lem:gadget-traversal}, points $r_1$, $r_1'$, $r_2$, and $r_2'$ cannot be adjacent to points of a turn gadget other than those labeled $p_1$ and $q_1$. 
We now show that if $S$ admits a circumscribing polygon, then we may assume that it contains the edges $r_1r_1'$ and $r_2r_2'$. 
Suppose, for contradiction, that there is a circumscribing polygon $P$, but $r_1r_1'\notin E(P)$.
Let $s_1$ and $s_1'$, respectively, be the neighbors of $r_1$ and $r_1'$ in $S$ and consider the adjacent turn gadget as shown in Fig.~\ref{fig:disjoint-seg-full}(a).
Since $r_1r_1'\notin E(P)$, vertex $r_1'$ ($r_1$)  must be adjacent to at least one vertex in $\{p_1,p_2,q_1,q_2,s_1\}$ ($\{p_1,p_2,q_1,q_2,s_1'\}$). We distinguish cases:
\begin{itemize}
    \item If $r_1'q_2\in E(P)$ or $r_1's_1\in E(P)$, then this edge blocks visibility for $r_1$ to all possible neighbors (other than $s_1$), and $r_1$ cannot have two neighbors in $P$. 
    \item If $r_1'q_1\in E(P)$, then this edge would block all but one visibility edge for $r_1$, and would imply $r_1q_2\in E(P)$. Then the only possible second neighbor of $r_1$ in $P$ is $s_1$. In this case, we can modify $P$ by expanding $r_1'q_1$ to the path $(r_1',r_1,q_1)$ and shortcutting the path $(s_1,r_1,q_2)$ to $s_1q_2$.
    \item If $r_1'p_1\in E(P)$, then the right side of $p_1r_1'$ is in the interior (exterior) of $P$ because of Lemma~\ref{lem:gadget-traversal} applied for the turn gadget between at $p_1$ ($s_1'$), which is a contradiction.  
\end{itemize}    
The symmetric claims also hold for neighbors of $r_1$. By exclusion, the only remaining possibility is that $P$ contains the paths $(s_1',r_1',p_2)$ and $(s_1,r_1,q_2)$.
However, since $P$ must visit $p_1$, this implies $p_2p_1\in E(P)$.
Then the path $(s_1',r_1'p_2,p_1)$ in $P$ contradicts Lemma~\ref{lem:gadget-traversal} for vertices $s_1'$ and $p_1$, as the interior of $P$ should be on the same side of every edge of $P$. We have shown that if $P$ is a circumscribing polygon for $S$, then $r_1' r_1\in E(P)$. By symmetry the same holds for $r_2$ and $r_2'$, completing the proof that the reflex angle gadgets behave as the corners of a polygon in any circumscribing polygon.
\end{proof}

As already discussed in the descriptions of the gadgets, the reduction can be computed in polynomial time and by Lemma~\ref{lem:equivalence-circ-disjoint} we have the following result.

\begin{theorem}\label{thm:hardness-circ-disjoint}
It is NP-complete to decide whether a set of disjoint line segments admits a circumscribing polygon.
\end{theorem}

\section{Simple Polygonizations of Disjoint Segments}
\label{sec:hardness}

Rappaport~\cite{rappaport1989computing} proved that it is NP-hard to decide whether a given PSLG $G=(V,E)$ admits a polygonization: The reduction from \textsc{Hamiltonian Path in Planar Cubic Graphs} (HPPCG) produces an instance in which $G$ is a union of disjoint paths, each edge in $E$ is horizontal or vertical, and the vertices in $V$ have integer coordinates, bounded by a polynomial in $n=|V|$. Rappaport~\cite{rappaport1989computing} raised the question whether the problem remains NP-hard when $G$ is a perfect matching (i.e., a set of disjoint line segments in the plane). 
In this section, we settle this problem (Theorem~\ref{thm:hardness}). 
Specifically, we modify Rappaport's reduction, and describe the \emph{connection gadget}  made of disjoint line segments that simulates a pair of adjacent line segments.

\paragraph{Description of the Connection Gadget.}
We later provide the details on how to construct the gadget.
Refer to Fig.~\ref{fig:polygonize}.
Given a PSLG $G=(V,E)$, with a vertex $p_2\in V$ of degree 2,
incident to $p_1p_2,p_2p_3\in E$, delete the edge $p_1p_2$, and insert
6 new edges $p_1p_2'$, $p_4p_5$, $p_6p_7$, $p_8p_9$, $p_{10}p_{11}$, $p_{12}p_{13}$,
and 11 new vertices $p_2'$ and $p_i$ ($i=3,\ldots  , 13$).
Denote by $G'=(V',E')$ the resulting new PSLG.
We choose the position of the new vertices close to $p_2$ so that: 
\begin{enumerate}
\item[{\rm (i)}] the two small segments $p_6p_7$ and $p_{10}p_{11}$ are visible only from points $p_4, p_5, p_8$, and $p_9, p_{12}, p_{13}$ respectively;
%(ii) $p_2'$ (resp., $p_2$) cannot see any vertex in the interior of the wedge defined by $\overrightarrow{p_2'p_1}$ and $\overrightarrow{p_2'p_8}$ (resp., $\overrightarrow{p_2p_9}$ and $\overrightarrow{p_2p_8}$);
\item[{\rm (ii)}] the union of the visibility regions of $p_4$, $p_5$, $p_{12}$, and $p_{13}$ each contain only some of $p_2, p_2', p_4, \ldots,$ $p_{13}$, but no other vertices.
\end{enumerate}

\begin{figure}[htbp]
	\centering
	\includegraphics[width=.8\linewidth]{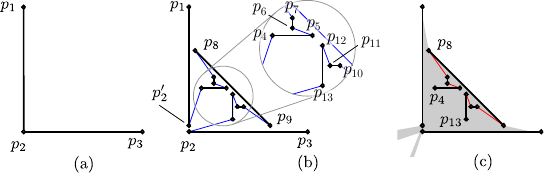}
	\caption{(a) Two adjacent axis-parallel line segments $p_1p_2$ and $p_2p_3$.
(b) Connection gadget that simulates (a) using seven disjoint line segments. The polygonal path shown in black and blue line segments is $[p_1,p_2',p_4,p_5,p_6,p_7,p_8,p_9,p_{10},p_{11},p_{12},p_{13},p_2,p_3]$.
(c) The union of the visibility regions of the solid black points $p_4$ and $p_{13}$.}
	\label{fig:polygonize}
\end{figure}

\paragraph{Construction of the Connection Gadget.} 
We now describe the placement of the points of the connection gadget that ensures (i)--(ii).
%We first define the rectangle $B=B(G)$ as the rectangle obtained by scaling the axis-aligned bounding box of a PSLG $G$ by 2, anchored at the center of the bounding box.
Let $\delta$ be the feature ratio of $G$, defined as the minimum distance between a vertex and a nonincident edge of $G$.
Let $\alpha$ be the smallest angle between two adjacent edges in $G$, and set $\eps=\delta \sin^{-1}(\frac{\alpha}{4n})$, where $n$ is the number of vertices in $G$.
We describe now the position of auxiliary lines.
Refer to Fig.~\ref{fig:reduction}.
The dashed line $\ell_1$ connects the pair of points in $p_1p_2p_3$ that are $\eps$ apart from $p_2$. This line will contain vertices $p_9$ and $p_8$.
Let $W$ be the cone of angle $\alpha$ from $p_2$ placed so that its bisector coincides with the bisector of $\angle p_1p_2p_3$.
Divide $W$ into $2n$ cones with equal angle.
By definition of $\eps$, if we place a disk of radius $\eps$ at each vertex of $G$, each disk can intersect only two $\frac{\alpha}{2n}$ cones. By the pigeonhole principle, one of these cones $w$ will contain no such disks.
The visibility of $p_2$ inside $w$ is either empty or the interior of an edge containing points that are more than $\eps$ apart from the endpoints.
Let $\ell_2$ and $\ell_3$ be the supporting lines of the rays defining $w$, and let $\ell_4$ be the bisector of $w$.
Let $a$ (resp. $b$) be the intersection point between $\ell_1$ and $\ell_4$ (resp., $\ell_2$), and $c$ be the intersection point of the circle of radius $\delta$ centered at $p_2$ and $\ell_3$ that is in the boundary of $w$.
%Let $\ell_5$ be the line through $a$ and $c$.
We place vertex $p_2'$ at the intersection of $ac$ and $p_1p_2$.
Let $t$ be the intersection point of $\ell_2$ and $ac$, and let $d$ (resp., $e$) be the point in $p_1p_2$ (resp., $p_2p_3$) that is $2\eps$ away from $p_2$.
%Place $p_8$ (resp., $p_9$) either at the intersection between $\ell_1$ and $dt$ (resp., $ed_2'$) or so that property (ii) is satisfied, whichever is closer to $p_1p_2$ (resp., $p_2p_3$).
Place $p_8$ (resp., $p_9$) at the intersection between $\ell_1$ and $dt$ (resp., $ed_2'$).
We will place points $p_4,\ldots,p_7$ and $p_{10},\ldots,p_{13}$ in the triangle $\Delta(a b t)$.
This guarantees that (ii) is satisfied.

\begin{figure}[htbp]
	\centering
	\includegraphics[width=.8\linewidth]{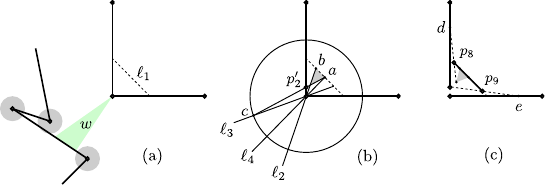}
	\caption{The figure is not drawn scale.
(a) The construction of the dashed line $\ell_1$ and the cone~$w$.
(b) The placement of  $p_2'$, and the triangle $\Delta (abt)$.
(c) The placement of $p_8$ and $p_9$.}
	\label{fig:reduction}
\end{figure}

%Refer to Fig.~\ref{fig:reduction}(X!!!!).
Let $f$ be the midpoint of $ab$, and $p_4'p_5'$ (resp., $p_{12}'p_{13}'$) be the segment defined by the intersection between $\Delta(a b t)$ and the line through $f$ parallel to $p_2p_3$ (resp., $p_1p_2$).
Define the segment $p_4p_5$ (resp., $p_{12}p_{13}$) by scaling $p_4'p_5'$ (resp., $p_{12}'p_{13}'$) by $\frac{1}{2}$ anchored at its midpoint.
Let $p_6'p_7'$ (resp., $p_{10}'p_{11}'$) be the segment contained in the perpendicular bisector of $p_4p_5$ (resp., $p_{12}p_{13}$) that is only visible to $p_4$, $p_5$, and $p_8$ (resp., $p_9$, $p_{12}$, and $p_{13}$).
Define the segment $p_6p_7$ (resp., $p_{10}p_{11}$) by scaling $p_6'p_7'$ (resp., $p_{10}'p_{11}'$) by $\frac{1}{2}$ anchored at its midpoint.
The construction guarantees that (i) is satisfied.
This concludes the construction of the connection gadget.

% \begin{lemma}\label{lem:connection-gadget1}
% Let $G$ be a PSLG with a connection gadget described above.
% Assume that a circumscribing polygon $P$ of $G$ contains $p_1p_2'$ as an edge and that the interior of $P$ is to the left of $\overrightarrow{p_1p_2}$.
% Then, $P$ contains a chain from $p_2'$ to $p_2$ containing exactly 11 edges going only through vertices $p_4,\ldots,p_{13}$.
% The statement is also true by replacing $p_1p_2'$ with $p_2p_3$.
% \end{lemma}
% \begin{proof}
% We first show that $P$
% \end{proof}

% \begin{lemma}\label{lem:connection-gadget}
% Let $G$ a PSLG containing a path $p_1p_2p_3$ such that $\text{degree}(p_2)=2$, $p_3$ is to the left of $\overrightarrow{p_1p_2}$, and let $G'$ be the PSLG obtained from $G$ by replacing $p_1p_2p_3$ by a connector gadget.
% Assume that if a circumscribing simple polygon of $G$ (resp., $G'$) exists, it contains $p_1p_2$ (resp., $p_1p_2'$) as an edge and the interior of the polygon is to the left of this edge.
% Then, $G$ admits a circumscribing simple polygon if and only if $G'$ also does.
% Moreover, such polygon contain the chain $p_1$
% \end{lemma}

\begin{theorem}\label{thm:hardness}
It is NP-complete to decide whether a set $S$ of disjoint line segments admits a simple polygonization,
even if $S$ contains only segments with 4 distinct slopes.
\end{theorem}

\begin{proof}
Membership in NP is proven in~\cite{rappaport1989computing}.
We reduce NP-hardness from finding polygonizations for a disjoint union of paths.
Let $G=(V,E)$ be a PSLG produced by the reduction in~\cite{rappaport1989computing}.  
Let $n=|V|$.
We modify $G$ by simultaneously replacing every vertex of degree 2 by a connection gadget
(described above), and show that the resulting plane straight-line
matching $M$ admits a polygonization if and only if $G$ does.
Since each gadget is constructed independently, all coordinates can be described by bounded-degree rational functions as each is obtained by a constant number of intersections between lines and circles determined by $G$.
Since $G$ contains only axis-parallel edges, the construction produces a plane straight-line matching $M$, in which all edges have up to four distinct slopes.
The reduction clearly runs in time polynomial in $n$.

We now show that $M$ admits a polygonization if and only if $G$ does.
Note that the connection gadget places edges in the convex corner of a degree-2 vertex in $G$, and it does not block or create visibility between two leafs of $G$.
By construction, if $p$ is a leaf in $G$, then the set of other leaves visible from $p$ remains same in $M$.
Since $G$ is max-degree-2, it remains to prove that, for every connection gadget, a polygonization of $M$ must contain a chain of length 11 from $p_2'$ to $p_2$ that uses only edges of the connection gadget.

By (i) of the connection gadget, if a simple polygonization $P$ of $M$ exists, $P$ must connect $p_8$ with $p_6$ or $p_7$, and $p_6$ or $p_7$ to $p_4$ or $p_5$; otherwise, $P$ would contain a cycle of length 4 and $P$ would be disconnected.
The same argument applies to vertices $p_9,\ldots,p_{13}$.
Fig.~\ref{fig:polygonize}(c) shows the forced edges in a polygonization in red.
By (ii), $p_2'$ must be adjacent to $p_4$ or $p_5$, and $p_2$ must be adjacent to $p_{12}$ or $p_{13}$, or else either $P$ would contain a cycle of length 10 and $P$ would be disconnected, or $P$ would not be simple.
Hence, each gadget behaves exactly like a degree-2 vertex of $G$, and $M$ admits a polygonization if and only if $G$ does.
\end{proof}

\section{Conclusions}
\label{sec:con}

Our results raise interesting open problems, among others, about circumscribing polygons in the plane (Section~\ref{ssec:3slopes}), and about higher dimensional generalizations (Section~\ref{ssec:3space}).

\subsection{Geometric Matching with Few Slopes}
\label{ssec:3slopes}

As noted above, Gr\"unbaum~\cite{Grunbaum94} constructed an arrangement of 6 disjoint segments in $\mathbb{R}^2$ that does not admit a circumscribing polygon; see Fig.~\ref{fig:diode}. If all segments have the same slope (but they are not all collinear), then there always exists a circumscribing polygon. We conjecture that disjoint segments with two distinct slopes still admit a circumscribing polygon. Here we present negative instances with three slopes.

\begin{proposition}\label{pro:threeslopes}
For every $n\geq 9$, there is a set of $n$ disjoint segments of 3 different slopes
that do not admit a circumscribing polygon.
\end{proposition}
\begin{proof}
Consider the set $S=\{s_1,\ldots ,s_9\}$ of segments in Fig.~\ref{fig:3slopes}, where $s_i=a_ib_i$, for $i=1,\ldots ,9$. Suppose, for the sake of contradiction, that $P$ is a circumscribing polygon for $S$. Rappaport, Imai, and Toussaint~\cite[Lemma~2.1]{Rappaport1990} proved that in every Hamiltonian simple polygon, the vertices of $\conv(S)$ appear in the same counterclockwise order in $P$ and $\conv(S)$. Hence vertices $a_1$, $a_2$, $b_2$, and $b_1$ appear in this ccw order in $P$. Denote by $A$ and $B$, respectively, the polygonal path from $a_1$ to $a_2$, and from $b_2$ to $b_1$ in ccw order along $P$. Since $a_1b_1$ and $a_2b_2$ are edges of $\conv(S)$, they are also edges of $P$.

\begin{figure}[htbp]
	\centering
	\includegraphics[width=.5\linewidth]{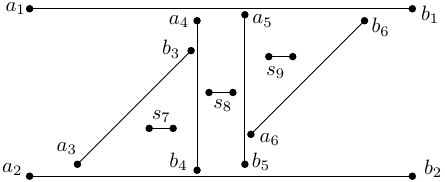}
	\caption{A set of 9 disjoint line segments of slopes $0$, $1$, and $\infty$, that do not admit a circumscribing polygon.}
	\label{fig:3slopes}
\end{figure}

Consequently, the endpoints of $s_3,\ldots, s_7$ are in $A$ or $B$. By symmetry, we may assume that at least one endpoint of $s_8$ is in $A$. Since $a_1a_2$ is also an edge of $\conv(S)$, it does not cross any edge of $A$, and so it can complete the path $A$ into a simple polygon. By \cite[Lemma~2.1]{Rappaport1990}, the vertices of $\conv(A)$ appear in the same order in both $\conv(A)$ and $A$.

Note first that $A$ contains some vertex to the right of $s_8$; otherwise, either $A$ crosses segment $a_4b_4$,
or both $a_4$ and $b_4$ are vertices of $P$ and so $a_4b_4$ is an external diagonal of $P$.
This means that $A$ contains some vertices of segments $s_5,s_6,s_9$ (as $b_1$ and $b_2$ are in $B$).
Note that all possible edges between the vertices to the left of $a_5b_5$ and the endpoints of $s_6$ and $s_9$
cross $a_5b_5$. If $A$ contains any vertex to the right of $a_5b_5$, then $A$ would
contain both $a_5$ and $b_5$, and $a_5b_5$ would be an external diagonal of $P$.
Consequently $A$ contains $a_5$ or $b_5$, but no vertices to the right of $a_5b_5$.
If $A$ contains $b_5$ without $a_5$, then either $A$ crosses $a_4b_4$ or
$a_4b_4$ is an external diagonal of $P$. We conclude that $A$ contains $a_5$.

It follows that $B$ contains $b_1,b_2$, both endpoints of $s_6$ and $s_9$, and possibly $b_5$.
Points $b_1$, $b_2$, $a_6$, $b_6$, and an endpoint of $s_9$ are vertices of $\conv(B)$.
Since they appear in the same ccw order in $\conv(B)$ and $B$, segment $a_6b_6$ is an external diagonal of $P$.
This contradicts our initial assumption and completes the proof.
\end{proof}

\subsection{Higher Dimensions}
\label{ssec:3space}

Generalizations to higher dimensions are also of interest. For a set $V$ of points in $\mathbb{R}^3$, a \emph{polyhedralization} is a polyhedron homotopic to a sphere whose vertex set is $V$.\footnote{We thank Joe Mitchell for introducing us to the high dimensional variations of this problem.}
It is known that every set of $n\geq 4$ points in general position admits a polyhedralization~\cite{AgarwalHTT08}, and even a polyhedralization of bounded vertex degree~\cite{BarequetBCDDILSSTW13}.

For a set $S$ of disjoint line segments in $\mathbb{R}^3$, we define a \emph{polyhedralization} as a polyhedron homotopic to a ball whose vertices are the segment endpoints, and every segment in $S$ is either an edge or an (external or internal) diagonal. A polyhedralization \emph{circumscribes} $S$ if every segment in $S$ is an edge or an internal diagonal.
It is not difficult to see that an arrangement of disjoint segments in general position in $\mathbb{R}^3$ need not admit a polyhedralization.

\begin{proposition}\label{pro:polyhedronization3D}
For every $n\geq 4$, there is a set of $n$ disjoint segments in $\mathbb{R}^3$
that do not admit a polyhedralization.
\end{proposition}
\begin{figure}[htbp]
	\centering
	\includegraphics[width=.55\linewidth]{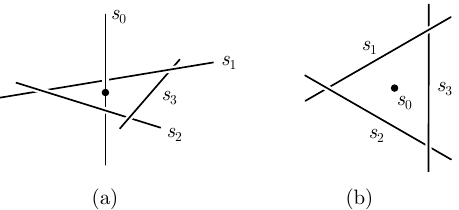}
	\caption{Left: A set of 4 disjoint segments in $\mathbb{R}^3$ that do not admit a polyhedralization; perspective view (a) and view from above (b).}
	\label{fig:3D}
\end{figure}
\begin{proof}
Let $s_0=[(0,0,-1),(0,0,1)]$ be a vertical segment along the $z$-axis. We construct $s_1,\ldots ,s_{n-1}$ by taking line segments along the supporting lines of a regular $(n-1)$-gon centered at the origin in the $xy$-plane, and then perturbing them to be in general position; see Fig.~\ref{fig:3D}.

Suppose that these $n$ segments admit a polyhedralization $P$. Then $s_0$ is the edge of some face. However, the triangle spanned by the two endpoints of $s_0$ and an endpoint of any segment $s_i$, $i\in \{1,\ldots, n-1\}$, stabs segment $s_{i-1}$ or $s_{i+1}$ (where the subindex is taken modulo $n-1$). This contradicts our assumption that every segment is an edge of $P$.
\end{proof}

We suspect that it is NP-hard to decide whether a set of $n$ line segments in $\mathbb{R}^3$
admits a polyhedralization (or even a circumscribing polyhedralization). However, our proof
techniques do not seem to extend to higher dimensions.

\subsection{Open Problems}
We conclude with a collection of open problems.
\begin{enumerate}%\itemsep-2pt
\item In Section~\ref{sec:hardness} we established NP-hardness for the polygonization problem, even when the input consists of disjoint line segments with four distinct slopes.
    Is it NP-hard to decide whether $n$ disjoint {\it axis-parallel} segments in the plane admit a polygonization?
    Is it NP-hard for segments of 3 possible directions?
%\item In Section~\ref{sec:hardness2} we proved that it is NP-complete to decide whether a 2-regular PSLG admits a circumscribing polygon. We do not know whether the problem remains hard for 1-regular PSLGs (i.e., disjoint line segments). The connection gadgets we designed for the polygonization problem (Section~\ref{sec:hardness2}) do not seem to work for circumscribing polygons. 
%Is it NP-hard to decide whether $n$ disjoint segments admit a circumscribing polygon?
%\csaba{I've removed the problem that we have solved.}
\item Does every arrangement of disjoint axis-parallel segments in $\mathbb{R}^2$, not all in a line, admit a circumscribing polygon?
\item Does every arrangement of disjoint line segments in $\mathbb{R}^3$, not all in a plane, admit a circumscribing polyhedron?
\item We can decide in $O(n\log n)$ time whether $n$ disjoint segments are extensible to disjoint rays (cf.~Section~\ref{sec:rays}). Can we decide efficiently whether they admit escape routes?
\item Let $f(n)$ be the maximum integer such that every set of $n$ disjoint segments contains $f(n)$ segments that admit a circumscribing polygon. In Section~\ref{sec:circumscribe}, we prove a lower bound of $f(n)=\Omega(\sqrt{n})$. Is it possible that $f(n)=\Omega(n)$? Is there a nontrivial upper bound?
\item Let $g(n)$ be the maximum integer such that every set of $n$ disjoint segments contains $g(n)$ segments that are extensible to disjoint rays. Theorem~\ref{thm:escape} implies $g(n)\leq f(n)$. Pach and Rivera-Campo~\cite{PachR98} proved that $g(n)=\Omega(n^{1/3})$, and Lemma~\ref{lem:lowerbound} gives $g(n)=O(\sqrt{n})$.
    What is the asymptotic growth rate of $g(n)$?
\item Let $h(n)$ be the maximum integer such that every set of $n$ disjoint segments contains $h(n)$ segments that admit an escape route. Theorem~\ref{thm:escape} implies $g(n)\leq h(n)\leq f(n)$. We have $h(n)=\Omega(n^{1/3})$ and $h(n)=O(\sqrt{n})$. What is the asymptotic growth rate of $h(n)$?
\end{enumerate}

%\small
\bibliographystyle{plainurl}
\bibliography{main}

\begin{thebibliography}{10}

\bibitem{AgarwalHTT08}
Pankaj~K. Agarwal, Ferran Hurtado, Godfried~T. Toussaint, and Joan Trias.
\newblock On polyhedra induced by point sets in space.
\newblock {\em Discrete Applied Mathematics}, 156(1):42--54, 2008.
\newblock \href {https://doi.org/10.1016/j.dam.2007.08.033}
  {\path{doi:10.1016/j.dam.2007.08.033}}.

\bibitem{BarequetBCDDILSSTW13}
Gill Barequet, Nadia Benbernou, David Charlton, Erik~D. Demaine, Martin~L.
  Demaine, Mashhood Ishaque, Anna Lubiw, Andr{\'{e}} Schulz, Diane~L. Souvaine,
  Godfried~T. Toussaint, and Andrew Winslow.
\newblock Bounded-degree polyhedronization of point sets.
\newblock {\em Comput. Geom.}, 46(2):148--153, 2013.
\newblock \href {https://doi.org/10.1016/j.comgeo.2012.02.008}
  {\path{doi:10.1016/j.comgeo.2012.02.008}}.

\bibitem{CardinalHKTW18}
Jean Cardinal, Michael Hoffmann, Vincent Kusters, Csaba~D. T{\'{o}}th, and
  Manuel Wettstein.
\newblock Arc diagrams, flip distances, and {H}amiltonian triangulations.
\newblock {\em Comput. Geom.}, 68:206--225, 2018.
\newblock \href {https://doi.org/10.1016/j.comgeo.2017.06.001}
  {\path{doi:10.1016/j.comgeo.2017.06.001}}.

\bibitem{ChambersEGL12}
Erin~W. Chambers, David Eppstein, Michael~T. Goodrich, and Maarten
  L{\"{o}}ffler.
\newblock Drawing graphs in the plane with a prescribed outer face and
  polynomial area.
\newblock {\em J. Graph Algorithms Appl.}, 16(2):243--259, 2012.
\newblock \href {https://doi.org/10.7155/jgaa.00257}
  {\path{doi:10.7155/jgaa.00257}}.

\bibitem{CHIBA1989}
Norishige Chiba and Takao Nishizeki.
\newblock The {H}amiltonian cycle problem is linear-time solvable for
  4-connected planar graphs.
\newblock {\em Journal of Algorithms}, 10(2):187 -- 211, 1989.
\newblock \href {https://doi.org/10.1016/0196-6774(89)90012-6}
  {\path{doi:10.1016/0196-6774(89)90012-6}}.

\bibitem{DL10}
Emilio Di~Giacomo and Giuseppe Liotta.
\newblock The {H}amiltonian augmentation problem and its applications to graph
  drawing.
\newblock In {\em Proc. 4th International Workshop Algorithms and Computation
  (WALCOM)}, volume 5942 of {\em LNCS}, pages 35--46, Berlin, 2010. Springer.
\newblock \href {https://doi.org/10.1007/978-3-642-11440-3_4}
  {\path{doi:10.1007/978-3-642-11440-3_4}}.

\bibitem{garcia2000lower}
Alfredo Garc{\'{\i}}a, Marc Noy, and Javier Tejel.
\newblock Lower bounds on the number of crossing-free subgraphs of
  {K\({}_{\mbox{N}}\)}.
\newblock {\em Comput. Geom.}, 16(4):211--221, 2000.
\newblock \href {https://doi.org/10.1016/S0925-7721(00)00010-9}
  {\path{doi:10.1016/S0925-7721(00)00010-9}}.

\bibitem{GareyJT76}
Michael~R. Garey, David~S. Johnson, and Robert~E. Tarjan.
\newblock The planar {H}amiltonian circuit problem is {NP}-complete.
\newblock {\em {SIAM} J. Comput.}, 5(4):704--714, 1976.
\newblock \href {https://doi.org/10.1137/0205049} {\path{doi:10.1137/0205049}}.

\bibitem{Grunbaum94}
Branko Gr\"unbaum.
\newblock Hamiltonian polygons and polyhedra.
\newblock {\em Geombinatorics}, 3:83--89, 1994.

\bibitem{HoffmannT03}
Michael Hoffmann and Csaba~D. T{\'{o}}th.
\newblock Segment endpoint visibility graphs are {H}amiltonian.
\newblock {\em Comput. Geom.}, 26(1):47--68, 2003.
\newblock \href {https://doi.org/10.1016/S0925-7721(02)00172-4}
  {\path{doi:10.1016/S0925-7721(02)00172-4}}.

\bibitem{hurtado2013plane}
Ferran Hurtado and Csaba~D. T{\'o}th.
\newblock Plane geometric graph augmentation: A generic perspective.
\newblock In J{\'a}nos Pach, editor, {\em Thirty Essays on Geometric Graph
  Theory}, pages 327--354. Springer, New York, 2013.
\newblock \href {https://doi.org/10.1007/978-1-4614-0110-0_17}
  {\path{doi:10.1007/978-1-4614-0110-0_17}}.

\bibitem{IshaqueST13}
Mashhood Ishaque, Diane~L. Souvaine, and Csaba~D. T{\'{o}}th.
\newblock Disjoint compatible geometric matchings.
\newblock {\em Discrete {\&} Computational Geometry}, 49(1):89--131, 2013.
\newblock \href {https://doi.org/10.1007/s00454-012-9466-9}
  {\path{doi:10.1007/s00454-012-9466-9}}.

\bibitem{Mirzaian92}
Andranik Mirzaian.
\newblock Hamiltonian triangulations and circumscribing polygons of disjoint
  line segments.
\newblock {\em Comput. Geom.}, 2:15--30, 1992.
\newblock \href {https://doi.org/10.1016/0925-7721(92)90018-N}
  {\path{doi:10.1016/0925-7721(92)90018-N}}.

\bibitem{ORourkeR94}
Joseph O'Rourke and Jennifer Rippel.
\newblock Two segment classes with {H}amiltonian visibility graphs.
\newblock {\em Comput. Geom.}, 4:209--218, 1994.
\newblock \href {https://doi.org/10.1016/0925-7721(94)90019-1}
  {\path{doi:10.1016/0925-7721(94)90019-1}}.

\bibitem{OZEKI2018}
Kenta Ozeki, Nico~Van Cleemput, and Carol~T. Zamfirescu.
\newblock Hamiltonian properties of polyhedra with few 3-cuts---{a} survey.
\newblock {\em Discrete Mathematics}, 341(9):2646--2660, 2018.
\newblock \href {https://doi.org/10.1016/j.disc.2018.06.015}
  {\path{doi:10.1016/j.disc.2018.06.015}}.

\bibitem{PachR98}
J{\'{a}}nos Pach and Eduardo Rivera{-}Campo.
\newblock On circumscribing polygons for line segments.
\newblock {\em Comput. Geom.}, 10(2):121--124, 1998.
\newblock \href {https://doi.org/10.1016/S0925-7721(97)00023-0}
  {\path{doi:10.1016/S0925-7721(97)00023-0}}.

\bibitem{rappaport1989computing}
David Rappaport.
\newblock Computing simple circuits from a set of line segments is
  {NP}-complete.
\newblock {\em SIAM Journal on Computing}, 18(6):1128--1139, 1989.
\newblock \href {https://doi.org/10.1137/0218075} {\path{doi:10.1137/0218075}}.

\bibitem{Rappaport1990}
David Rappaport, Hiroshi Imai, and Godfried~T. Toussaint.
\newblock Computing simple circuits from a set of line segments.
\newblock {\em Discrete {\&} Computational Geometry}, 5(3):289--304, 1990.
\newblock \href {https://doi.org/10.1007/BF02187791}
  {\path{doi:10.1007/BF02187791}}.

\bibitem{sharir2013counting}
Micha Sharir, Adam Sheffer, and Emo Welzl.
\newblock Counting plane graphs: Perfect matchings, spanning cycles, and
  {K}asteleyn's technique.
\newblock {\em J. Comb. Theory, Ser. {A}}, 120(4):777--794, 2013.
\newblock \href {https://doi.org/10.1016/j.jcta.2013.01.002}
  {\path{doi:10.1016/j.jcta.2013.01.002}}.

\bibitem{Tutte60}
William~T. Tutte.
\newblock Convex representations of graphs.
\newblock {\em Proc. London Math. Soc.}, 10(38):304--320, 1960.
\newblock \href {https://doi.org/10.1112/plms/s3-10.1.304}
  {\path{doi:10.1112/plms/s3-10.1.304}}.

\bibitem{UrabeW92}
Masatsugu Urabe and Mamoru Watanabe.
\newblock On a counterexample to a conjecture of {M}irzaian.
\newblock {\em Comput. Geom.}, 2:51--53, 1992.
\newblock \href {https://doi.org/10.1016/0925-7721(92)90020-S}
  {\path{doi:10.1016/0925-7721(92)90020-S}}.

\bibitem{W31}
Hassler Whitney.
\newblock A theorem on graphs.
\newblock {\em Annals of Mathematics, 2nd Ser.}, 32(2):378--390, 1931.
\newblock \href {https://doi.org/10.1007/978-1-4612-2972-8_2}
  {\path{doi:10.1007/978-1-4612-2972-8_2}}.

\end{thebibliography}

%\appendix

\end{document}